\documentclass[12pt]{iopart}

\usepackage{iopams}
\usepackage{amsthm}
\usepackage{amssymb}
\usepackage{bbold}
\usepackage{color}
\newtheorem{theorem}{Theorem}
\newtheorem{proposition}{Proposition}
\newtheorem{lemma}{Lemma}
\newtheorem{remarks}{Remarks}

\newtheorem{corollary}{Corollary}
\theoremstyle{definition}

\begin{document}
\small
\title[Renormalization of $\phi_4^4$ theory on the half space 
$\mathbb{R}^+ \times\mathbb{R}^3$ with flow equations]
{Perturbative renormalization of $\phi_4^4$ theory 
on the half space $\mathbb{R}^+ \times\mathbb{R}^3$ with flow equations}
\author{Majdouline BORJI \footnote{majdouline.borji@polytechnique.edu}, Christoph KOPPER \footnote{christoph.kopper@polytechnique.edu}}

\address{Centre de Physique Théorique CPHT, CNRS, UMR 7644}
\address{Institut Polytechnique de Paris, 91128 Palaiseau, France}
\vspace{10pt}
\begin{indented}
\item[]April 2022
\end{indented}

\begin{abstract}
In this paper, we give a rigorous proof of the renormalizability of the 
massive  $\phi_4^4$ theory on a half-space, using the renormalization group 
flow equations. We find that five counter-terms are needed to make the 
theory finite, namely $\phi^2$, $\phi\partial_z\phi$, $\phi\partial_z^2\phi$, 
$\phi\Delta_x\phi$ and $\phi^4$ for $(z,x)\in\mathbb{R}^+\times\mathbb{R}^3$. 
The amputated correlation functions are distributions in position space. 
We consider a suitable class of test functions and prove inductive bounds 
for the correlation functions folded with these test functions. The bounds 
are uniform in the cutoff and thus directly lead to renormalizability.
\end{abstract} 

\section{Introduction}
\indent The renormalization of quantum field theories that break translation invariance is of great importance since many interesting quantum field theories break this symmetry.  One may ask if the renormalizability of a given theory depends only on the interaction introduced and the dimension of space-time, or whether it depends also on the geometrical and topological properties of the space-time. In a previous work \cite{25}, we studied the $\phi^4_4$-theory on a lattice which is a regularization scheme that breaks translation invariance, and we found that the theory is renormalizable and the euclidean symmetries are restored for the renormalized correlation functions. In \cite{9}, the authors considered the breaking of translation invariance by studying the $\phi^4_4$ interaction on a Riemanian manifold and found that it is renormalizable. Only one additional counter-term which renormalizes the curvature is needed, compared to the $\phi^4_4$-theory in the euclidean space-time. In this work, we are interested in the renormalizability of the $\phi_4^4$ scalar field theory on a space with a boundary, which is another manifestation of the breaking of translation invariance. \\
 \indent A simple model to study surface effects in quantum field theory is the semi-infinite scalar field model which first appeared in 1971 \cite{2}. It is defined starting from the massive $\phi_4^4$ model in infinite space, with the difference that it is defined on a half space bounded by a plane. In this model, three types of boundary conditions are considered in the litterature, namely Dirichlet, Neumann and Robin boundary conditions (b.c.). From a mathematical point of view, each b.c. corresponds to a self-adjoint extension of the Laplacian in $\mathbb{R}^+\times \mathbb{R}^3$. The self-adjointness of the Laplacian is normally required in order to be able to define the propagator of a quantum field theory. Each boundary condition defines a particular propagator. \\
\indent Lubensky and Rubin \cite{17,18} studied a model of ferromagnetically coupled classical spins on a semi-infinite lattice. Using a mean-field approach they provided a qualitatively correct understanding of the different phases undergone by the system which are: the ordinary, extraordinary, surface and special transitions. The phenomenological theory of scaling \cite{22,24} was generalized to surfaces, and it implied relations between bulk critical exponents and the additional surface critical exponents, needed to describe the singular behaviour of surface related properties. However, renormalization of the model is necessary when wants to go beyond the mean-field approximation.\\
Diehl and Dietrich \cite{10,12} studied the critical behavior of the semi-infinite system using renormalization group methods.
They considered the ordinary \cite{10} and special transitions \cite{12} which correspond respectively to the Dirichlet and Robin boundary conditions and found that in addition to the usual two bulk counter-terms, an additional surface counter-term is needed to make the two-point function finite in the case of the Dirichlet boundary condition. For the Robin boundary condition, two surface counter-terms are needed. The calculations were performed to two-loop order using dimensional regularization, and the surface counter-terms were obtained by inserting the operators $\lim_{z\rightarrow 0}\partial_z\phi(z,x)$ where 
$x\in \mathbb{R}^3$ in the case of Dirichlet b.c., and $\phi(0,x)$, 
$\phi^2(0,x)$ in the case of the Robin b.c., 
where $\phi$ is the considered scalar field. \\
\indent The semi-infinite model was also adressed by Symanzik in his study of the Schrödinger representation for renormalizable quantum fields \cite{1} in which he discusses the renormalization of surface operators in a different context, but also finds that surface counter-terms are required to make the two-point function finite.\\ 
\indent In \cite{15}, Albuquerque calculated the one-loop two-point function using a cut-off regularization in the case of the Robin boundary condition. In addition to the usual mass counter-term, two additional counter-terms are needed to make the (non-amputated) tadpole finite. They depend on whether the external points are on the surface or not. If none of them is on the surface, then only one surface counter-term is required, and it diverges linearly in the cutoff. However, if at least one of the external points lies on the surface, then in addition to the linearly diverging counter-term, an extra surface counter-term is needed, and it diverges logarithmically with the cutoff. These findings suggest that the renormalization of the amputated and unamputated diagrams is different.\\
\indent In this paper, we give a rigorous proof of the renormalizability of the $\phi_4^4$ massive semi-infinite model using the renormalization group flow equations. The paper is organized as follows: In Section 2 we present the semi-infinite model scalar field theory with all the possible boundary conditions and their associated propagators. We also present the properties of the flowing propagator and the associated heat kernel together with the considered action and the system of perturbative flow equations satisfied by the connected amputated Schwinger distributions (CAS). Section 3 will be devoted to prove some regularity properties of the support of the gaussian measure associated to the regularized propagator. To establish bounds on the CAS, which are distributions, they have to be folded first with test functions. In Section 4 a suitable class of test functions is introduced, together with tree structures that will be used in the bounds on the CAS to be derived. In Section 5 we state the boundary and the renormalization conditions used to integrate the flow equations of the irrelevant and relevant terms respectively. Section 6 is the central one of this paper. We state and prove inductive bounds on the amputated Schwinger distributions folded with the introduced test functions which, being uniform in the cutoff, directly lead to renormalizability. 
\section{The Action and the flow equations}
\subsection{The half space and the possible boundary conditions}
We consider the half space or what we call also the semi-infinite space $\mathbb{R}^+\times\mathbb{R}^3$ where $\mathbb{R}^+:=\left[0,\infty\right)$ and $\mathbb{R}^{+*}:=\left(0,\infty\right)$.\\
Let $\mathcal{C}^{\infty}_o(\mathbb{R}^+\times\mathbb{R}^3)$ be the space of compactly supported smooth functions defined on the considered half space. For $k\in \mathbb{N}^*$ we consider the Hilbert Sobolev spaces $H^{k}:=W^{k,2}$. We denote by $H^{1}_0\left(\mathbb{R}^+\right)$ the set of functions in $H^{1}\left(\mathbb{R}^+\right)$ that vanish at the boundary $0$. The Laplacian $\Delta$ defined on the Hilbert space $L^2\left(\mathbb{R}^+\times\mathbb{R}^3\right)$ has the following self-adjoint extensions:
\begin{itemize}
    \item The Dirichlet Laplacian $\Delta_D$ defined by 
    $$\forall u \in \mathcal{C}^{\infty}_o(\mathbb{R}^+\times\mathbb{R}^3),\qquad \Delta_Du=\Delta u$$
with the following domain
\begin{eqnarray*}
\fl D\left(\Delta_D\right):=\left\{u\in L^2\left(\mathbb{R}^+\times\mathbb{R}^3\right)\left.\right|u(z,\cdot)
\in H^2(\mathbb{R}^3)~\forall z\geq 0;\right.\\\left.
u(\cdot,x)\in H^{1}_0\left(\mathbb{R}^+\right)\cap H^2\left(\mathbb{R}^+\right)~\forall x \in \mathbb{R}^3\right\}~.
\end{eqnarray*}
    \item The Neumann Laplacian $\Delta_N$ defined by 
    $$\forall u \in \mathcal{C}^{\infty}_o(\mathbb{R}^+\times\mathbb{R}^3),\qquad \Delta_Nu=\Delta u$$
with the following domain
\begin{eqnarray*}
\fl D\left(\Delta_N\right):=\left\{u\in L^2\left(\mathbb{R}^+\times\mathbb{R}^3\right)\left.\right|u(z,\cdot)
\in H^2(\mathbb{R}^3)~\forall z\geq 0;\right.\\
\left.u(\cdot,x)\in H^2\left(\mathbb{R}^+\right),
\partial_zu(z,x)|_{z=0}=0~~\forall x \in \mathbb{R}^3\right\}~.
\end{eqnarray*}
    \item The Robin Laplacian $\Delta_R$ defined by 
    $$\forall u \in \mathcal{C}^{\infty}_o(\mathbb{R}^+\times\mathbb{R}^3),\qquad \Delta_Ru=\Delta u$$
with the domain
\begin{eqnarray*}
\fl D\left(\Delta_R\right):=\left\{u\in L^2\left(\mathbb{R}^+\times\mathbb{R}^3\right)\left.\right|u(z,\cdot)
\in H^2(\mathbb{R}^3)~\forall z\geq 0;u(\cdot,x)
\in H^2\left(\mathbb{R}^+\right),\right.\\
\left.\partial_zu(z,x)|_{z=0}=cu(0,x)~~\forall x \in \mathbb{R}^3\right\},
~~c>0~.
\end{eqnarray*}
Each self-adjoint extension corresponds to a possible boundary condition. The massive propagators associated to these boundary conditions are defined from functional calculus by \begin{equation}\label{bla}
    C_{\bullet}\left((z,x);(z',x')\right)
=\int_0^{\infty}d\lambda~e^{-\lambda \left(-\Delta_{\bullet}+m^2\right)}
\left((z,x);(z',x')\right),
\end{equation}
where $\bullet\in \left\{D,R,N\right\}$ for respectively Dirichlet, Robin and Neumann boundary conditions. (\ref{bla}) can be written in terms of the heat kernels as 
\begin{eqnarray*}
    \fl \indent C_D\left((z,x);(z',x')\right)
&=\int_0^{\infty}d \lambda\ e^{-\lambda m^2}p_{B}\left(\lambda;x,x'\right)\frac{1}{\sqrt{2\pi \lambda}}\left(e^{-\frac{(z-z')^2}{2\lambda}}-e^{-\frac{(z+z')^2}{2\lambda}}\right),\nonumber\\
    \fl \indent C_N\left((z,x);(z',x')\right)
&=\int_0^{\infty}d \lambda\ e^{-\lambda m^2}p_{B}\left(\lambda;x,x'\right)\frac{1}{\sqrt{2\pi \lambda}}\left(e^{-\frac{(z-z')^2}{2\lambda}}+e^{-\frac{(z+z')^2}{2\lambda}}\right),\nonumber\\
    \fl \indent C_R\left((z,x);(z',x')\right)
&=\int_0^{\infty}d \lambda\ e^{-\lambda m^2}p_{B}\left(\lambda;x,x'\right)p_R\left(\lambda;z,z'\right),\nonumber
\end{eqnarray*}
where 
\begin{equation}\label{pb}
p_B(\lambda;x,x'):=\frac{1}{(2\pi \lambda)^\frac{3}{2}}e^{-\frac{(x-x')^2}{2\lambda}}
\end{equation}
and
\begin{equation}\label{2prime}
p_R(\lambda;z,z'):=
p_N(\lambda;z,z')-2\int_0^{\infty}\frac{dw}{\sqrt{2\pi \lambda}} 
\ e^{-w}\, e^{-\frac{\left(z+z'+\frac{w}{c}\right)^2}{2\lambda}}~.
\end{equation}
Here $p_N$ denotes the one-dimensional Neumann heat kernel
\begin{equation}\label{pneumann}
p_N(\lambda;z,z'):=\frac{1}{\sqrt{2\pi \lambda}}\left(\frac{e^{-\frac{(z-z')^2}{2\lambda}}+e^{-\frac{(z+z')^2}{2\lambda}}}{2}\right).
\end{equation}
In the $pz$-representation, which corresponds to taking the partial Fourier transformation with respect to the variable $x\in\mathbb{R}^3$, the Dirichlet, Neumann and Robin propagators simply read
\begin{eqnarray}
\fl~~~~C_D(p;z,z')=\frac{1}{2\sqrt{p^2+m^2}}\left[e^{-\sqrt{p^2+m^2}\left|z-z'\right|}-e^{-\sqrt{p^2+m^2}\left|z+z'\right|}\right]
\ ,
\end{eqnarray}
\begin{eqnarray}
\fl~~~~C_N(p;z,z')=\frac{1}{2\sqrt{p^2+m^2}}\left[e^{-\sqrt{p^2+m^2}\left|z-z'\right|}+e^{-\sqrt{p^2+m^2}\left|z+z'\right|}\right]\ ,
\end{eqnarray}
\begin{eqnarray}\label{Robpro}
\fl~~~~C_R(p;z,z')=\frac{1}{2\sqrt{p^2+m^2}}\left[e^{-\sqrt{p^2+m^2}\left|z-z'\right|}+\frac{\sqrt{p^2+m^2}-c}{\sqrt{p^2+m^2}+c}
\ e^{-\sqrt{p^2+m^2}\left|z+z'\right|}\right]\ .
\end{eqnarray}
The Dirichlet boundary condition corresponds to $c\rightarrow \infty$ and the Neumann boundary condition to $c=0$. We study the Robin boundary condition since the other two conditions are limit cases of the former. One can easily verify that we have 
\begin{eqnarray}
 \fl C_D(p;0,z')=C_D(p;z,0)=0\ ,~~\lim_{z\rightarrow 0}~\partial_zC_N(p;z,z')
=\lim_{z'\rightarrow 0}~\partial_{z'}C_N(p;z,z')=0\ ,\\
 \fl ~~~~\lim_{z\rightarrow 0}~\partial_zC_R(p;z,z')=c~C_R(p;0,z')\ ,
\ ~\lim_{z'\rightarrow 0}~\partial_{z'}C_R(p;z,z')=c~C_R(p;z,0),\label{rob}
\end{eqnarray}
where we used that the associated 
heat kernels verify respectively 
the Dirichlet, Neumann and Robin boundary conditions.
\end{itemize}
\subsection{$\phi_4^4$ scalar field theory on the semi-infinite space}
We will analyze the perturbative renormalizability of the semi-infinite 
$\phi_4^4$ theory with Robin boundary conditions. It will be proved by 
analyzing the generating functional $L^{\Lambda,\Lambda_0}$ of connected 
amputated Schwinger distributions (CAS). The upper indices 
$\Lambda_0$ and $\Lambda$ enter through the regularized propagator. 
We choose the following regularization 
\begin{equation}
    C^{\Lambda,\Lambda_0}_R(p;z,z')=
\int_{\frac{1}{\Lambda_0^2}}^{\frac{1}{\Lambda^2}}
d\lambda\  e^{-\lambda(p^2+m^2)}p_R(\lambda;z,z')\ .
\end{equation}
Clearly, $C^{\Lambda,\Lambda_0}_R$ verifies the b.c. (\ref{rob}). For $\Lambda\rightarrow 0$ and $\Lambda_0\rightarrow \infty$ we recover the unregularized 
propagator (\ref{Robpro}).
We denote  
\begin{equation}\label{p3}
\dot{C}_{R}^{\Lambda}(p;z,z')=\frac{\partial}{\partial\Lambda}{C}_{R}^{\Lambda,\Lambda_0}(p;z,z')=\dot{C}^{\Lambda}(p)\ p_R(\frac{1}{\Lambda^2};z,z')~,
\end{equation}
where $\dot{C}^{\Lambda}(p)=-\frac{2}{\Lambda^3}e^{-\frac{p^2+m^2}{\Lambda^2}}$.\\
The starting point in writing an euclidean quantum field theory is to define the associated path integral given by the corresponding gaussian measure. We assume $0\leq \Lambda\leq \Lambda_0<\infty$ so that the flow parameter $\Lambda$ takes the role of an infrared cutoff, whereas $\Lambda_0$ is a UV cutoff. The full propagator is recovered for $\Lambda=0$ and $\Lambda_0\rightarrow\infty$. For finite $\Lambda_0$ and in finite volume the positivity and the regularity properties of $C^{\Lambda,\Lambda_0}_R$ permit to define the theory rigorously from the functional integral 
\begin{eqnarray}\label{fl}
 e^{-\frac{1}{\hbar}\left(L^{\Lambda,\Lambda_0}(\phi)+I^{\Lambda,\Lambda_0}\right)}:&=
\int d\mu_{\Lambda,\Lambda_0,R}(\Phi)\  
e^{-\frac{1}{\hbar}L^{\Lambda_0,\Lambda_0}(\Phi+\phi)}\ ,\\
    L^{\Lambda,\Lambda_0}(0)&=0\ ,\nonumber
\end{eqnarray}
where the factors of $\hbar$ have been introduced to allow for a consistent loop expansion in the sequel. Here, $d\mu_{\Lambda,\Lambda_0,R}$ denotes the Gaussian measure with covariance $\hbar C_R^{\Lambda,\Lambda_0}$. The test functions $\phi$ and
 $\Phi$  are supposed to be in the support of the Gaussian measure $d\mu_{\Lambda,\Lambda_0,R}$, which in particular implies that they are in $\mathcal{C}^{\infty}\left(\mathbb{R}^+\times\mathbb{R}^3\right)$ as we will prove in Section 3. The normalization factor $e^{-\frac{1}{\hbar}I^{\Lambda,\Lambda_0}}$ is due to vacuum contributions. It diverges in infinite volume so that we can take the infinite volume limit only when it has been eliminated \cite{16}. We do not make the finite volume explicit here since it plays no role in the sequel. \\
\indent The functional $L^{\Lambda_0,\Lambda_0}(\phi)$ is the bare interaction of a renormalizable theory including counter-terms, viewed as a formal power series in $\hbar$. For shortness we will pose in the following, with $z\in \mathbb{R}^+$, $p\in \mathbb{R}^3$ and $x\in\mathbb{R}^3$,
$$\int_{z}:=\int_0^{\infty}dz~,~~~~\int_p:=
\int_{\mathbb{R}^3}\frac{d^3p}{(2\pi)^3}~,~~~~\int_S:=\int_{\mathbb{R}^3}d^3x~,~~~~\int_V:=
\int_0^{\infty}dz\int_{\mathbb{R}^3}d^3x\ .$$
Since translation invariance is broken in the $z$-direction (the semi-line), all counter-terms may be $z$-dependent. In general, the constraints on the bare action result from the symmetry properties of the theory which are imposed, on its field content and on the form of the propagator. It is therefore natural to consider the general bare interaction
\begin{eqnarray}\label{bare}
    \fl L^{\Lambda_0,\Lambda_0}(\phi)=\frac{\lambda}{4!}\int_{V}\phi^4(z,x)
+\frac{1}{2}\int_V \biggl(a^{\Lambda_0}(z)\phi^2(z,x)
-b^{\Lambda_0}(z)\phi(z,x)\Delta_x \phi(z,x)\\
-d^{\Lambda_0}(z)\phi(z,x)\partial_z^2\phi(z,x) 
+s^{\Lambda_0}(z)\phi(z,x)(\partial_z\phi)(z,x)
+\frac{2}{4!}c^{\Lambda_0}(z)\phi^4(z,x)\biggr)
\ .\nonumber
\end{eqnarray}
Here we supposed the theory to be symmetric under $\phi\rightarrow-\phi\,$, 
and we included only relevant terms with respect to (\ref{fl}) in the sense of the renormalization group.
The functions $a^{\Lambda_0}(z),~b^{\Lambda_0}(z),~c^{\Lambda_0}(z),~d^{\Lambda_0}(z)$ and $s^{\Lambda_0}(z)$ are supposed to be smooth. \\
The flow equation (FE) is obtained from (\ref{fl}) on differentiating w.r.t. $\Lambda$. For the steps of the computation, we refer the reader to \cite{16} and \cite{18}. It is a differential equation for the functional $L^{\Lambda,\Lambda_0}$:\\
\begin{equation}\label{floEq}
    \fl~~~~~\partial_{\Lambda}(L^{\Lambda,\Lambda_0}+I^{\Lambda,\Lambda_0})=\frac{\hbar}{2}\langle \frac{\delta}{\delta \phi},\dot{C}_R^{\Lambda}\,\frac{\delta}{\delta \phi}\rangle L^{\Lambda,\Lambda_0}-\frac{1}{2}\langle\frac{\delta}{\delta \phi}L^{\Lambda,\Lambda_0},\dot{C}_R^{\Lambda}\,\frac{\delta}{\delta \phi}L^{\Lambda,\Lambda_0}\rangle\ .
\end{equation}
By $\langle,\rangle$ we denote the standard inner product in 
$L^2(\mathbb{R}^+\times\mathbb{R}^3)$.\\
We may expand the functional 
$L^{\Lambda,\Lambda_0}(\phi)$ in a formal power series w.r.t. $\hbar$,
\begin{equation*}
    L^{\Lambda,\Lambda_0}(\phi)=\sum_{l=0}^{\infty}\hbar^l L^{\Lambda,\Lambda_0}_l(\phi)\ .
\end{equation*}
Corresponding expansions for $a^{\Lambda_0}(z),~b^{\Lambda_0}(z)$..., are $a^{\Lambda_0}(z)=\sum_{l=0}^{\infty}\hbar^l a^{\Lambda_0}_l(z) $, etc. From $L^{\Lambda,\Lambda_0}_l(\phi)$ we obtain the CAS distributions of loop order $l$ as 
\begin{equation*}
    \mathcal{L}^{\Lambda,\Lambda_0}_{l,n}\left((z_1,x_1),\cdots,(z_n,x_n)\right)
:=\delta_{\phi(z_1,x_1)}\cdots \delta_{\phi(z_n,x_n)}L^{\Lambda,\Lambda_0}_l|_{\phi=0}~,
\end{equation*}
where we used the notation $\delta_{\phi(z,x)}=\delta/\delta \phi(z,x)\,$. \\
Since translation invariance in the $x$-directions is preserved, we will use in all what follows a mixed representation, where the Fourier transform to $p$-space is performed only with respect to $x \in \mathbb{R}^3$. In this representation, we set 
\begin{eqnarray}\label{13*}
\fl\mathcal{L}_{l,n}^{\Lambda,\Lambda_0}\left(z_1;\vec{p}_n;\Phi_{n}\right)=\int_0^{\infty}dz_2\cdots dz_n~\mathcal{L}_{l,n}^{\Lambda,\Lambda_0}
\left((z_1,p_1),\cdots,(z_n,p_n)\right)\phi_2(z_2)\cdots\phi_n(z_n)\ .
\end{eqnarray}
Here we denote 
 $$\Phi_{n}(z_2,\cdots,z_n):=\prod_{i=2}^n\phi_i(z_i),~~~\vec{p}_n:=\left(p_1,\cdots,p_n\right)\in\mathbb{R}^{3n},~~\left\|\vec{p}_n\right\|:
=\sup_{1\leq i\leq n}|p_i|$$
and 
$$\delta^{(3)}(p_1+\cdots+p_n)
\mathcal{L}_{l,n}^{\Lambda,\Lambda_0}\left((z_1,p_1),\cdots,(z_n,p_n)\right)
=(2\pi)^{3(n-1)}\frac{\delta^n}{\delta \phi(z_1,p_1)
\cdots\delta \phi(z_n,p_n)}L^{\Lambda,\Lambda_0}_l(\phi)|_{\phi\equiv0}~.$$
The $\delta^{(3)}(p_1+\cdots+p_n)$ appears because of the translation
invariance in the $x$ directions.
The FE for the CAS distributions derived from (\ref{floEq}) are \cite{9}-\cite{18}
\begin{eqnarray}\label{FEL}
    \fl\partial_{\Lambda}\partial^w\mathcal{L}_{l,n}^{\Lambda,\Lambda_0}\left((z_1,p_1),\cdots,(z_n,p_n)\right)\\
\fl=\frac{1}{2}\int_z \int_{z'}\int_k 
\partial^w\mathcal{L}_{l-1,n+2}^{\Lambda,\Lambda_0}
\left((z_1,p_1),\cdots,(z_n,p_n),(z,k),(z',-k)\right)
\dot{C}^{\Lambda}_R(k;z,z')\nonumber\\
   \fl \indent -\frac{1}{2}\int_z \int_{z'} \sum_{l_1,l_2}'\sum_{n_1,n_2}'\sum_{w_i}c_{w_i} \left[\partial^{w_1}\mathcal{L}_{l_1,n_1+1}^{\Lambda,\Lambda_0}((z_1,p_1),\cdots,(z_{n_1}p_{n_1}),(z,p))\partial^{w_3}\dot{C}^{\Lambda}_R(p;z,z')\right.\nonumber\\\left.\times\ 
\partial^{w_2}\mathcal{L}_{l_2,n_2+1}^{\Lambda,\Lambda_0}((z',-p),\cdots,(z_{n},p_{n}))\right]_{rsym},\nonumber\\
    p=-p_1-\cdots-p_{n_1}=p_{n_1+1}+\cdots+p_n\ .
\nonumber
\end{eqnarray}
Here we wrote (\ref{FEL}) directly in a form where a number $|w|$ of momentum derivatives, characterized by a multi-index, act on both sides and we used the shorthand notation
\begin{eqnarray}
   \fl \partial^{w}:=\prod^{n}_{i=1}\prod_{\mu=0}^3\left(\frac{\partial}{\partial p_{i,\mu}}\right)^{w_{i,\mu}}~ \mathrm{with}~ w=(w_{1,0},\cdots,w_{n,3}),~
|w|=\sum{w_{i,\mu}},~w_{i,\mu}\in \mathbb{N}^*\ .
\end{eqnarray}{}
The symbol "rsym" means summation over those permutations of the momenta $(z_1,p_1)$, $\cdots$ ,$(z_n,p_n)$, which do not leave invariant the (unordered) subsets $\left((z_1,p_1),\cdots,(z_{n_1},p_{n_1})\right)$ and $\left((z_{n_1+1},p_{n_1+1})\right.$,$\left.\cdots,(z_n,p_n)\right)$, and therefore, produce mutually different pairs of (unordered) image subsets, and the primes restrict the summations to $n_1+n_2=n$, $l_1+l_2=l$, $w_1+w_2+w_3=w$, respectively. The combinatorial factor $c_{\left \{w_i\right \}}=w!(w_1!w_2!w_3!)^{-1}$ stems from Leibniz's rule. In the loop order $l=0$, the first term on the RHS is absent.\\
 \section{Regularity of the support of the regularized gaussian measure}
The bare interaction $L^{\Lambda_0,\Lambda_0}$ is composed of powers of the field $\phi$ and of its derivatives. It can not be given any mathematical meaning if the field $\phi$ is not sufficiently regular, e.g. in $\mathcal{C}^2\left(\mathbb{R}^+\times\mathbb{R}^3 \right)$. In this section, we prove that the field $\phi$ belongs to $\mathcal{C}^{\infty}(\mathbb{R}^+\times\mathbb{R}^3)\,$. This result
is due to the regularity properties of the support of the gaussian measure $\mu_{\Lambda,\Lambda_0,R}\,$.\\
Since the theory which we study is massive, we consider the UV-regularized propagator without infrared cut-off
\begin{equation*}
    C^{\Lambda_0}_R(p;x,y)=\int_{\frac{1}{\Lambda_0^2}}^{\infty}d\lambda~e^{-\lambda(p^2+m^2)}p_R(\lambda;x,y)
\end{equation*}
assuming that $\Lambda_0\geq 1$, $p\in \mathbb{R}^3$ and $x,~y\in \mathbb{R}^+$.\\
The same arguments work for the Gaussian measure associated to the propagator $C_R^{\Lambda,\Lambda_0}(p;x,y)$. We prove the following 
\begin{proposition}\label{prop1}
Let $\mu_{\Lambda_0,R}$ be the Gaussian measure associated to the propagator $C_R^{\Lambda_0}$. The support of $\mu_{\Lambda_0,R}$ satisfies by 
\begin{equation*}
    \mathrm{supp}\mu_{\Lambda_0,R}\subset \bigcap_{n\geq 1} \left\{\left(-\Delta_R+m^2\right)^{-n}{L}^{2}\left(\mathbb{R}^+\times\mathbb{R}^3\right)\right\}.
\end{equation*}
\end{proposition}
\noindent Before stating the proof of Proposition \ref{prop1}, we recall the following corollary of the Minlos theorem \cite{3,4}.
\begin{corollary}\label{coco}
Given a nuclear space $E$, $\mu$ a measure on $E'$, and $C$ its characteristic function, we introduce a continuous inner product $\left(\cdot,\cdot\right)_0$ on $E$ and let $H_0$ be the completion of $E$ with respect to $\left(\cdot,\cdot\right)_0$. Suppose that $C$ is continuous on $H_0$. Let $T$ be a Hilbert-Schmidt operator on $H_0$ satisfying:
\begin{itemize}
    \item[(a)] $T$ is one to one (injective map).
    \item[(b)] $E\subset \mathrm{Im}T$ and $T^{-1}(E)$ is dense in $H_0\,$.
    \item[(c)] The map $T^{-1}:E\rightarrow H_0$ is continuous.
\end{itemize}
Then the support of $\mu$ is on $(T^{-1})^*H_0'\subset E'$.
The notations $(T^{-1})^*$ and $H_0'$ are used for the "adjoint" and the "dual space" in the pairing between $E$ and $E'$. 
\end{corollary}
\noindent For a proof of this corollary see \cite{3,4}. 
\begin{proof}
We apply Corollary \ref{coco} to $E=\mathcal{S}(\mathbb{R}^+\times\mathbb{R}^3)\cong\mathcal{S}(\mathbb{R}^+)\otimes \mathcal{S}(\mathbb{R}^3)\,$. 
This is a nuclear space which is a tensor product of two nuclear spaces. See \cite{5} for the proof that $\mathcal{S}(\mathbb{R}^+)$ is a nuclear space. The theorem A.4.1 in \cite{26} implies the existence of the Gaussian measure $\mu_{\Lambda_0,R}$ 
with covariance $C_R^{\Lambda_0}$ with support included in $\mathcal{S}'(\mathbb{R}^+\times\mathbb{R}^3)$. We apply the corollary of Minlos's theorem to the scalar product $\langle f,g \rangle_n:=\langle f , P^{-2n}g\rangle$ where 
$\langle\,,\rangle$ is the usual scalar product in $L^2\left(\mathbb{R}^+\times\mathbb{R}^3\right)$, and $P=-\Delta_R+m^2$. $P^n$ is a unitary map from $L^2(\mathbb{R}^+\times\mathbb{R}^3)$ into $H_{-n}$, the completion of $\mathcal{S}(\mathbb{R}^+\times\mathbb{R}^3)$ with respect to $\langle\,,\rangle_n\,$. We verify first that the regularized covariance is continuous on $H_{-n}$ for any $n \in \mathbb{N}$, that is 
    \begin{equation*}
        \exists C>0~\mathrm{such~that~}\forall f, g \in H_{-n}:~~~~\left |\langle f,C_R^{\Lambda_0}g\rangle\right|\leq C\|f\|_{H_{-n}}\|g\|_{H_{-n}}~.
    \end{equation*}
One can verify that the operators $C^{\Lambda_0}_R$ and $\left(-\Delta_R+m^2\right)^{-n}$ commute. Since $\left(-\Delta_R+m^2\right)^{-n}$ is self-adjoint, we obtain
\begin{equation}\label{amounti}
    \langle f,C_R^{\Lambda_0}g\rangle=\left\langle \left(-\Delta_R+m^2\right)^{-n}f,\left(-\Delta_R+m^2\right)^{2n}C^{\Lambda_0}_R\left(-\Delta_R+m^2\right)^{-n}g\right\rangle.
\end{equation}
By the Cauchy-Schwarz inequality we obtain 
\begin{eqnarray}\label{**}
   \fl \left|\langle f,C_R^{\Lambda_0}g\rangle\right|
\leq \int \frac{d^3 p}{(2\pi)^3}\left\|\left(-\Delta_R+m^2\right)^{-n}f\right\|_{L^2(\mathbb{R}^+)}(p)\nonumber\\
\times\
\left\|\left(-\Delta_R+m^2\right)^{2n}C^{\Lambda_0}_R\left(-\Delta_R+m^2\right)^{-n} g\right\|_{L^2(\mathbb{R}^+)}(p)~,
\end{eqnarray}
where 
\begin{eqnarray}\label{star}
    \fl\left\|\left(-\Delta_R+m^2\right)^{2n}C^{\Lambda_0}_R\left(-\Delta_R+m^2\right)^{-n} g\right\|^2_{L^2(\mathbb{R}^+)}(p)\\\fl\indent =\int_0^{\infty}dx \left | \int_0^{\infty} dy \left(-\Delta_R+m^2\right)_x^{2n}C_R^{\Lambda_0}(p;x,y)\int_0^{\infty} dz\left(-\Delta_R+m^2\right)^{-n}(p;y,z)g(z,p)  \right|^2.\nonumber
\end{eqnarray}
Using again the Cauchy-Schwarz inequality we obtain 
\begin{eqnarray}
    \fl \left | \int_0^{\infty} dy \left(-\Delta_R+m^2\right)_x^{2n}C_R^{\Lambda_0}(p;x,y)\int_0^{\infty} dz\left(-\Delta_R+m^2\right)^{-n}(p;y,z)g(z,p)  \right|^2\\
    \leq \int_0^{\infty}dy \left |\left(-\Delta_R+m^2\right)_x^{2n}C_R^{\Lambda_0}(p;x,y)\right|^2 \left\|\left(-\Delta_R+m^2\right)^{-n}g\right\|^2_{L^2(\mathbb{R}^+)}.\nonumber
\end{eqnarray}
Therefore (\ref{star}) can be bounded by
\begin{equation}\label{leib1}
    \fl \indent \left(\int_0^{\infty} dx \int_0^{\infty}dy \left |\left(-\Delta_R+m^2\right)_x^{2n}C_R^{\Lambda_0}(p;x,y)\right|^2\right) \left\|\left(-\Delta_R+m^2\right)^{-n}g\right\|^2_{L^2(\mathbb{R}^+)}.
\end{equation}
We have
\begin{equation}\label{leib2}
   \fl \indent \left(-\Delta_R+m^2\right)_x^{2n}C_R^{\Lambda_0}(p;x,y)=\int_{\frac{1}{\Lambda^2_0}}^{\infty}d\lambda e^{-\lambda(p^2+m^2)}\left(-\partial_x^2+p^2+m^2\right)^{2n}p_R(\lambda;x,y). 
\end{equation}
Using the Leibniz formula we get
\begin{equation}\label{leibz0}
   \fl\indent  \left(-\partial_x^2+p^2+m^2\right)^{2n}p_R(\lambda;x,y)=\sum_{k=0}^{2n}C_{2n}^{k}(-1)^k(p^2+m^2)^{2n-k}\partial^{2k}_xp_R(\lambda;x,y)\ .
\end{equation}
One can prove by induction
\begin{equation}\label{leibz1}
 \partial^{2k}_x\left(\frac{1}{\sqrt{2\pi\lambda}}e^{-\frac{(x-y)^2}{2\lambda}}\right)=\lambda^{-k}\sum_{i=0}^k c_i(k) \left(\frac{x-y}{\sqrt{\lambda}}\right)^i \left(\frac{1}{\sqrt{2\pi\lambda}}e^{-\frac{(x-y)^2}{2\lambda}}\right)~,  
\end{equation}
where $c_i(k) \in \mathbb{R}$, which implies that
\begin{equation*}
    \left|\partial_x^{2k}p_B(\lambda;x,y)\right|\leq c_k~\lambda^{-k}p_B(2\lambda;x,y)~.
\end{equation*}
Here $c_k$ is a positive constant. Therefore we find using (\ref{2prime}) and (\ref{pneumann}) that
\begin{eqnarray*}
   \fl~~~~~~\left|\partial_x^{2k}p_R(\lambda;x,y)\right|\leq c_k~\lambda^{-k}\left(p_B(2\lambda;x,y)+p_B(2\lambda;x,-y)\right.\\\left.+2\int_0^{\infty}dw~ e^{-w}p_B\left(2\lambda;x,-\frac{w}{c}-y\right)\right).
\end{eqnarray*}
Using that $$2\int_0^{\infty}\frac{dw}{\sqrt{2\pi \lambda}}\ e^{-w}e^{-\frac{\left(x+y+\frac{w}{c}\right)^2}{2\lambda}}=p_N(\lambda;x,y)-p_R(\lambda;x,y)\ , $$
we obtain
\begin{eqnarray*}
   \fl\left|\partial_x^{2k}p_R(\lambda;x,y)\right|&\leq c_k~\lambda^{-k}\left(p_B(2\lambda;x,y)+p_B(2\lambda;x,-y)+p_N(2\lambda;x,y)-p_R(2\lambda;x,y)\right)\\
   &\leq 2c_k~\lambda^{-k} p_N(2\lambda;x,y),
\end{eqnarray*}
and this implies that
\begin{equation*}
    \fl{\left|\left(-\partial_x^2+p^2+m^2\right)^{2n}p_R(\lambda;x,y)\right|\leq 2C_n\sum_{k=0}^{2n}\lambda^{-k}(p^2+m^2)^{2n-k}p_N(2\lambda;x,y).}
\end{equation*}
where $C_n:=\sup_{1\leq k \leq n}c_k$.\\
Therefore we obtain
\begin{eqnarray*}
    \fl{\left |\left(-\Delta_R+m^2\right)^{2n}C_R^{\Lambda_0}(p;x,y)\right| 
    \leq \tilde{C}_n\Lambda_0^{4n} \,\sum_{k=0}^{2n}
\left(\frac{p^2+m^2}{\Lambda_0^2}\right)^{2n-k}
\int_{\frac{1}{\Lambda_0^2}}^{\infty}d\lambda\
 e^{-\lambda(p^2+m^2)} p_N(2\lambda;x,y)}\ ,
\end{eqnarray*}
where we used that 
$$p_R(\tau;x,y)\leq p_N(\tau;x,y)~~~~~~\forall \tau, x, y \in \mathbb{R}^+\ .$$
$C_n$ and $\tilde{C}_n$ are suitable
positive constants that depend on $n\,$.\\
We have by Cauchy-Schwarz
\begin{eqnarray*}
    \fl \left|\int_{\frac{1}{\Lambda^2_0}}^{\infty}d\lambda\ e^{-\lambda(p^2+m^2)}p_N(2\lambda;x,y)\right|^2 &\leq e^{-\frac{p^2+m^2}{\Lambda^2_0}}\int_{\frac{1}{\Lambda^2_0}}^{\infty}d\lambda e^{-{\lambda(p^2+m^2)}}|p_N(2\lambda;x,y)|^2\\
    &\leq  \frac{4\Lambda_0}{\sqrt{2\pi}} e^{-\frac{p^2+m^2}{\Lambda_0^2}}\int_0^{\infty}d\lambda\ e^{-\lambda(p^2+m^2)}\frac{1}{\sqrt{2\pi \lambda}}e^{-\frac{(x-y)^2}{2\lambda}}\\
    &\leq  \frac{4\Lambda_0}{\sqrt{2\pi}} e^{-\frac{p^2+m^2}{\Lambda^2_0}}\frac{e^{-\sqrt{p^2+m^2}|x-y|}}{\sqrt{p^2+m^2}}~,
\end{eqnarray*}
where we used that for $x,~y\in \mathbb{R}^+$
$$e^{-\frac{(x+y)^2}{2\lambda}}\leq e^{-\frac{(x-y)^2}{2\lambda}}~.$$
Therefore we have obtained the following bound for the first factor from (\ref{leib1})
\begin{eqnarray}\label{*}
    \fl~~\left(\int_0^{\infty} dx \int_0^{\infty}dy \left |\left(-\Delta_R+m^2\right)^{2n}C_R^{\Lambda_0}(p;x,y)\right|^2\right) \leq \Lambda_0^{8n+1}\mathcal{P}\left(\frac{p^2+m^2}{\Lambda_0^2}\right)\frac{e^{-\frac{p^2+m^2}{\Lambda_0^2}}}{\left(p^2+m^2\right)^{\frac{3}{2}}}~,
\end{eqnarray}
where $\mathcal{P}$ is a suitable polynomial with positive coefficients. All constants were absorbed in the polynomial $\mathcal{P}$, and we obtain the final bound for (\ref{amounti}) using again the Cauchy-Schwarz inequality w.r.t. the momenta $p\,$,
\begin{eqnarray}\label{***}
   \left|\langle f,C_R^{\Lambda_0}~g\rangle\right|
&\leq C_{\Lambda_0}\left\|\left(-\Delta_R+m^2\right)^{-n}f\right\|_{L^2(\mathbb{R}^+\times \mathbb{R}^3)}\left\|\left(-\Delta_R+m^2\right)^{-n} g\right\|_{L^2(\mathbb{R}^+\times \mathbb{R}^3)} \nonumber\\
   &=C_{\Lambda_0} \|f\|_{H_{-n}}\times\|g\|_{H_{-n}}~,
\end{eqnarray} 
The constant $C_{\Lambda_0}$ depends on $\Lambda_0$. To obtain (\ref{***}), we have performed the $p$-integral in (\ref{**}) using the bound (\ref{*}).
\item We apply Corollary \ref{coco} with $T=P^{-2}$. 
The operator $P^{-2}$ has the following kernel 
\begin{eqnarray*}
\fl\left(-\Delta_R+m^2\right)^{-2}\left(p;x,y\right)
:=\int_0^{\infty}\frac{du}{4(p^2+m^2)}
\left(e^{-\sqrt{p^2+m^2}|x-u|}+\frac{\sqrt{p^2+m^2}-c}
{\sqrt{p^2+m^2}+c}e^{-\sqrt{p^2+m^2}|x+u|}\right)\\
\times\ 
\left(e^{-\sqrt{p^2+m^2}|y-u|}
+\frac{\sqrt{p^2+m^2}-c}{\sqrt{p^2+m^2}+c}e^{-\sqrt{p^2+m^2}|y+u|}\right)\ .
\end{eqnarray*}
We can bound it as follows
\begin{eqnarray*}
\fl{\left(-\Delta_R+m^2\right)^{-2}\left(p;x,y\right)
\leq\int_0^{\infty}{du}~e^{-\sqrt{p^2+m^2}|x-u|}\ 
e^{-\sqrt{p^2+m^2}|y-u|}\frac{1}{\left(\sqrt{p^2+m^2}+c\right)^2}}~.
\end{eqnarray*}
$T$ is a Hilbert-Schmidt operator on $L^2(\mathbb{R}^+\times \mathbb{R}^3)$ since it is an integral operator with kernel in $L^2(\mathbb{R}^+\times \mathbb{R}^3)$. This is a consequence of
\begin{eqnarray}
  \fl \left\|\left(-\Delta_R+m^2\right)^{-2}\right\|_{L^2\left(\mathbb{R}^+\times \mathbb{R}^3\right)}^2\nonumber
   &\leq \int_{\mathbb{R}^3}\frac{d^3p}{(2\pi)^3}\frac{1}{\left(\sqrt{p^2+m^2}+c\right)^4}\left(\int_{0}^{\infty}dx~du~ e^{-2\sqrt{p^2+m^2}|x-u|}\right)^2  \nonumber\\&\leq C\int_{\mathbb{R}^3} \frac{d^3 p}{(2\pi)^3}\frac{1}{(p^2+m^2)^{4}}<\infty\ ~~~~~~~~~\mathrm{for~suitable~}C>0~.
\nonumber
\end{eqnarray}
$T$ satisfies the hypotheses of Corollary \ref{coco}. The dual of $H_{-n}$ is the space $H_n$ of functions whose image under $P^n$ is in $L^2$. Therefore, 
$\mu_{\Lambda_0,R}$ has support on the set $\left\{P^{2-n}f,~f\in L^2(\mathbb{R}^+\times \mathbb{R}^3)\right\}$. Since this is true for any $n\in \mathbb{N}^{*}$, we conclude that $\mu_{\Lambda_0,R}$ has its support within the set $ \bigcap_{n\geq 1} \left\{\left(-\Delta_R+m^2\right)^{-n}L^{2}\left(\mathbb{R}^+\times \mathbb{R}^3\right)\right\}\ $.
\end{proof}
\section{Test functions and Tree structures}
\subsection{Test functions}
Standard proofs of perturbative renormalizability by flow equations use inductive bounds on the $n$-point correlation functions. These objects are no more functions if considered in the mixed position-momentum space, but rather tempered distributions. We introduce tempered distributions in $\mathcal{S}'\left(\mathbb{R}^{+n}\right)$ w.r.t. the semi-norms $$\prod_{i=1}^n \mathcal{N}_{\alpha,\beta}\left(\phi_i\right)\ ,$$ 
where 
$\mathcal{N}_{\alpha,\beta}(\phi):=\sup_{0\leq \alpha,\beta\leq2}\left\|(1+z^{\beta})
\partial_z^{\alpha}\phi(z)\right\|_{\infty}$ 
and $\partial_z\phi|_{z=0}=\lim_{z\rightarrow 0^+}\partial_z\phi\,$. 
We refer the reader to \cite{5} for additional informations on the topological construction of $\mathcal{S}'\left(\mathbb{R}^+\right)$.\\ \indent 
We now introduce test functions against 
which  $\mathcal{L}_{l,n}^{\Lambda,\Lambda_0}$ will be integrated. In the sequel we will bound the CAS folded with 
 test functions of the following form:\\ Let $1\leq s \leq n$, we define 
 \begin{center}
 $\tau:=\inf \tau_{2,s}$ where $\tau_{2,s}=\left(\tau_2,\cdots,\tau_s\right)$ with $\tau_i>0$,
 \end{center}
 and similarly $z_{2,s}=\left(z_2,\cdots,z_s\right)$. Given $y_2,\cdots,y_s\in \mathbb{R}^+$, we define
\begin{equation}\label{phit}
    \phi_{\tau_{2,s},y_{2,s}}(z_{2,s}):=\prod_{i=2}^sp_R(\tau_i;z_i,y_i)
\prod_{i=s+1}^n \chi^{+}(z_i)\ ,
\end{equation}
where $\chi^{+}(z_i)$ is the characteristic function of the semi-line $\mathbb{R}^+$. This definition can be generalized by choosing any other subset of $s$ coordinates among $z_2,\cdots,z_n\, $. We also define for $2\leq j \leq s$
\begin{equation}\label{phij}
    \phi_{\tau_{2,s},y_{2,s}}^{(j)}(z_{2,n}):=p^{(1)}_R(\tau_j;z_j,z_1;y_j)\prod_{i=2,i\neq j}^s p_R(\tau_i;z_i,y_i)\prod_{i=s+1}^n \chi^{+}(z_i)
\end{equation}
with 
\begin{equation}\label{phij'}
    p^{(1)}_R(\tau_j;z_j,z_1;y_j)=p_R(\tau_j;z_j,y_j)-p_R(\tau_j;z_1,y_j)\ .
\end{equation}
\subsection{Tree structures}\label{ts}
We follow \cite{9} and  introduce the tree structures that will be used in establishing inductive bounds for the CAS. 
\begin{itemize}
\item[i)] We denote by $\mathcal{T}^s$ the set of all trees that have a root vertex and $s-1$ external vertices, where $s\geq 2\,$. 
Subsequently we will identify for shortness the vertices of the trees 
with a set of points in $\mathbb{R}^+\,$. 
For a tree $T^s \in \mathcal{T}^s$ we will call $z_1\in \mathbb{R}^+$ its root vertex, and $Y=\left\{y_2,\cdots,y_s\right\}$ the set of points in $\mathbb{R}^+$ to be identified with its external vertices. Likewise we call $z=\left\{z_1,\cdots,z_r\right\}$ with $r\geq 0$ the set of internal vertices of $T^s$. 
\item[ii)] We call $c_1=c(z_1)$ the incidence number of the root vertex, that is the number of the lines of the tree that have the root vertex as an edge.
The external vertices have incidence number $1\,$,
the internal vertices have incidence number $>1\,$.
 We call a line $p$ an external line of the tree if one of its edges is in 
$Y\,$. The set of external lines is denoted $\mathcal{J}\,$. 
The remaining lines are called internal lines of the tree and 
are denoted by $\mathcal{I}\,$.
\item[iii)] By $T^s_l$ we denote a tree $T^s \in \mathcal{T}^s$ satisfying $v_2+\delta_{c_1,1}\leq 3l-2+s/2$ for $l\geq 0$ and satisfying $v_2=0$ for $l=0$, where $v_n$ is the number of vertices having incidence number $n$. Then $\mathcal{T}^{s}_l$ denotes the set of all trees $T^{s}_l$. We indicate the external vertices and internal vertices of the tree by writing $T^{s}_l(z_1,y_{2,s},\vec{z})$ with $y_{2,s}=(y_2,\cdots,y_s)$ and $\vec{z}=(z_2,\cdots,z_{r+1})$.
\item[iv)] We also define for $i\leq s$ the set of twice rooted trees denoted as $\mathcal{T}^{s,(12)}_l$. The trees $T^{s,(12)}_l\in\mathcal{T}^{s,(12)}_l$ are defined exactly as the trees $T^s_l$ apart from the fact that they have two root vertices $z_1$ , $z_2$ with the property ii) above, and $s-2$ external vertices.
\item[v)] For a tree $T^{s+2}_l(z_1,y_{2,s+2},\vec{z})$ we define the reduced tree\\ $T^{s}_{l,y_i,y_j}(z_1,y_2,\cdots,y_{i-1},y_{i+1},\cdots,y_{j-1},y_{j+1},\cdots,y_{s+2},\vec{z}_{ij})$ to be the unique tree to be obtained from $T^{s+2}_l(z_1,y_{2,s},\vec{z})$ through the following procedure:
\begin{itemize}
    \item By taking off the two external vertices $y_i$, $y_j$ together with the external lines attached to them.
    \item By taking off the internal vertices -if any- which have acquired incidence number $c=1$ through the previous process, and by also taking off the lines attached to them.
    \item If a new vertex of incidence number $1$ is created, the second step of the process is repeated.
\end{itemize}
\end{itemize}
\subsection{Weight factors}
Let $0<\delta<1$. Given a set $\tau_{2,s}$ with $\tau:=\inf \tau_{2,s}$, a set of 
external vertices  
$y_{2,s}=\left\{y_2,\cdots,y_s\right\} \in (\mathbb{R}^+)^{s-1}\,$ 
and a set of internal vertices $\vec{z}=(z_2,\cdots,z_{r+1})
 \in (\mathbb{R}^+)^{r} ,$ and attributing positive parameters $\Lambda_{\mathcal{I}}=\left\{\Lambda_I|I\in \mathcal{I}\right\}$ to the internal lines, the weight factor $\mathcal{F}\left(\Lambda_{\mathcal{I}},\tau;T^s_l(z_1,y_{2,s},\vec{z})\right)$ of a tree $T^s_l(z_1,y_{2,s},\vec{z})$ at scales $\Lambda_I$ is defined as a product of heat kernels associated with the internal and external lines of the tree. We set
\begin{equation}\label{treeStr}
\mathcal{F}\left(\Lambda_{\mathcal{I}},\tau;T^s_l(z_1,y_{2,s},\vec{z})\right):=\prod_{I\in \mathcal{I}}p_B(\frac{1+\delta}{\Lambda_I^2};I)\prod_{J\in \mathcal{J}}p_B(\tau_{J,\delta};J)~,
\end{equation}
where $\tau_{J,\delta}$ denotes the entry $\tau_{i,\delta}$ in $\tau$ carrying the index of the external coordinate $y_i$ in which the external line $J$ ends, and $\tau_{i,\delta}:=(1+\delta) \tau_i$. For $I=\left\{a,b\right\}$ the notation $p_B(\frac{1+\delta}{\Lambda_I^2};I)$ stands for $p_B(\frac{1+\delta}{\Lambda_I^2};a,b)$. We also define the integrated weight factor 
\begin{equation}
\mathcal{F}\left(\Lambda,\tau;T^s_l;z_1,y_{2,s}\right):=\sup_{\Lambda\leq\Lambda_I\leq\Lambda_0}\int_{\vec {z}}\mathcal{F}\left(\Lambda_{\mathcal{I}},\tau;T^s_l(z_1,y_{2,s},\vec{z})\right).
\end{equation}
It depends on $\Lambda_0$, but note that its limit for $\Lambda_0\rightarrow \infty$ exists, and that typically the $\sup$ is expected to be taken for the minimal values of $\Lambda$ admitted. Therefore we suppress the dependence on $\Lambda_0$ in the notation. Finally we introduce the global weight factor $\mathcal{F}\left(\Lambda,\tau,z_1,y_{2,s}\right)$ or more shortly $\mathcal{F}_{s,l}^{\Lambda}(\tau)$ which is defined through
\begin{equation}\label{15}
     \mathcal{F}_{s,l}\left(\Lambda,\tau,z_1,y_{2,s}\right)
:=\sum_{T^s_l \in \mathcal{T}_l^s}\mathcal{F}
\left(\Lambda,\tau;T^s_l;z_1,y_{2,s}\right)\ .
\end{equation}
If this does not lead to ambiguity we write shortly 
\begin{equation}\label{1996}
    \mathcal{F}_{s,l}^{\Lambda}(\tau)\equiv 
\mathcal{F}_{s,l}\left(\Lambda,\tau,z_1,y_{2,s}\right)\ .
\end{equation}
In complete analogy we define the weight factors and global weight factors for twice rooted trees which we denote as $\mathcal{F}\left(\Lambda,\tau;T^{s,(12)}_l;z_1,z_2,y_{2,s}\right)$ resp. $\mathcal{F}_{s,l}^{(12)}\left(\Lambda,\tau,z_1,y_{2,s}\right)$ or $\mathcal{F}_{s,l}^{12}(\Lambda,\tau)\,$.\\
For $s=1$ we set $\mathcal{F}_{1,l}^{\Lambda}(\tau)\equiv 1$. This case corresponds to a tree $T^s_l$ with no external vertices.
\section{Boundary and renormalization conditions}
The relevant terms in the bare interaction are fixed by renormalization conditions at the value $\Lambda=0$ of the flow parameter, all other boundary terms are fixed at $\Lambda=\Lambda_0$. To extract the relevant terms contained in 
\begin{equation}
\mathcal{L}_{l,2}^{\Lambda,\Lambda_0}\left(z_1;0,0;\phi_2\right):=\int_{z_2}\mathcal{L}_{l,2}^{\Lambda,\Lambda_0}\left((z_1,0),(z_2,0)\right)\phi_2(z_2)
\end{equation}
and 
\begin{equation}
\mathcal{L}_{l,4}^{\Lambda,\Lambda_0}\left(z_1;0,0,0,0;\Phi_4\right):=\int_{z_{2,4}}\mathcal{L}_{l,4}^{\Lambda,\Lambda_0}\left((z_1,0),\cdots,(z_4,0)\right)\prod_{i=2}^4 \phi_i(z_i)\ ,
 \end{equation}
we use a Taylor expansion of the test functions $\phi_2$ and $\Phi_4$, which gives
\begin{eqnarray}\label{f1}
    \fl  \mathcal{L}_{l,2}^{\Lambda,\Lambda_0}\left(z_1;0,0;\phi_2\right)
    =a_l^{\Lambda,\Lambda_0}(z_1)\phi_2(z_1)-s_l^{\Lambda,\Lambda_0}(z_1)(\partial_{z_1}\phi_2)(z_1)+d_l^{\Lambda,\Lambda_0}(z_1) (\partial_{z_1}^2\phi_2)(z_1)
    \nonumber\\+l_{l,2}^{\Lambda,\Lambda_0}(z_1; \phi_2)\ ,
\end{eqnarray}
\begin{equation}\label{f2}
   \fl  \left(\partial_{p^2}\mathcal{L}_{l,2}^{\Lambda,\Lambda_0}\right)\left(z_1;0,0;\phi_2\right)=b_l^{\Lambda,\Lambda_0}(z_1)\phi_2(z_1)+\left(\partial_{p^2}l_{l,2}^{\Lambda,\Lambda_0}\right)(z_1; \phi_2)\ ,
\end{equation}
\begin{equation}\label{f3}
    \fl \mathcal{L}_{l,4}^{\Lambda,\Lambda_0}\left(z_1;0,\cdots,0;\Phi_4\right)=c_{l}^{\Lambda,\Lambda_0}(z_1)\phi_2(z_1)\phi_3(z_1)\phi_4(z_1)+l_{l,4}^{\Lambda,\Lambda_0}\left(z_1;\Phi_4\right)\ ,
\end{equation}
where $\Phi_4(z_2,z_3,z_4)=\prod_{i=2}^4\phi_i(z_i)\,$.\\
Then the relevant terms appear as 
\begin{eqnarray}\label{rel}
     a_l^{\Lambda,\Lambda_0}(z_1)=\int_0^{\infty} dz_2 \ \mathcal{L}_{l,2}^{\Lambda,\Lambda_0}\left((z_1,0),(z_2,0)\right),\\
    s_l^{\Lambda,\Lambda_0}(z_1)=\int_0^{\infty} dz_2 \, (z_1-z_2)\mathcal{L}_{l,2}^{\Lambda,\Lambda_0}\left((z_1,0),(z_2,0)\right),\\
    d_l^{\Lambda,\Lambda_0}(z_1)=\int_0^{\infty} dz_2 \, (z_1-z_2)^2\mathcal{L}_{l,2}^{\Lambda,\Lambda_0}\left((z_1,0),(z_2,0)\right),\\
    b_l^{\Lambda,\Lambda_0}(z_1)=\int_0^{\infty} dz_2\ \partial_{p^2}\left(\mathcal{L}_{l,2}^{\Lambda,\Lambda_0}\left((z_1,p),(z_2,-p)\right)\right)_{|_{p=0}},\\
     c_l^{\Lambda,\Lambda_0}(z_1)=\int_0^{\infty} dz_2dz_3dz_4 \ \mathcal{L}_{l,4}^{\Lambda,\Lambda_0}\left((z_1,0),\cdots,(z_4,0)\right)\label{lolo},
\end{eqnarray}
and the reminders $l_{l,2}^{\Lambda,\Lambda_0}\left(z_1; \phi_2\right)$, $\left(\partial_{p^2}l_{l,2}^{\Lambda,\Lambda_0}\right)\left(z_1;\phi_2\right)$ and $l_{l,4}^{\Lambda,\Lambda_0}\left(z_1;\Phi_4\right)$ have the form
\begin{eqnarray}
\fl\indent l_{l,2}^{\Lambda,\Lambda_0}(z_1;\phi_2)=\int_{0}^{\infty}dz_2\int_0^1 dt \frac{(1-t)^2}{2!}\partial_t^3{\phi_2}\left(t z_2+(1-t)z_1\right)\mathcal{L}_{l,2}^{\Lambda,\Lambda_0}((z_1;0),(z_2;0))\ ,
\end{eqnarray}
\begin{eqnarray}
\fl \left(\partial_{p^2}l_{l,2}^{\Lambda,\Lambda_0}\right)(z_1; \phi_2)
=\int_{0}^{\infty}\!\! dz_2\int_0^1 dt \frac{(1-t)^2}{2!}\partial_t^3{\phi_2}\left(t z_2+(1-t)z_1\right)\nonumber\\\times\partial_{p^2}\left(\mathcal{L}_{l,2}^{\Lambda,\Lambda_0}\left((z_1,p),(z_2,-p)\right)\right)_{|_{p=0}}\ ,
\end{eqnarray}
and
\begin{eqnarray}\label{loup}
\fl l_{l,4}^{\Lambda,\Lambda_0}(z_1;\Phi_4)\nonumber\\
\fl =\int_{0}^{\infty}dz_2dz_3dz_4~\mathcal{L}_{l,4}^{\Lambda,\Lambda_0}((z_1,0),\cdots,(z_4,0))\left[\int_0^1 dt~ {\partial_t\phi_2}\left(t z_2+(1-t)z_1\right)\phi_3(z_3)\phi_4(z_4)\right.\\\fl~~~~~ \left.
+\phi_2(z_1)\int_0^1 dt ~{\partial_t\phi_3}\left(tz_3+(1-t)z_1\right)\phi_4(z_4)+\phi_2(z_1)\phi_3(z_1)\int_0^1 dt ~{\partial_t\phi_4}\left(tz_4+(1-t)z_1\right)\right]~.\nonumber
\end{eqnarray}
\textbf{Boundary conditions at $\Lambda=\Lambda_0$}\\
The bare interaction implies that at $\Lambda=\Lambda_0$
\begin{eqnarray}\label{ii}
    \fl  \partial^{w}\mathcal{L}_{l,n}^{\Lambda_0,\Lambda_0}
\left((z_1,p_1),\cdots,(z_n,p_n)\right)=0~~~\forall n+|w|\geq 5,~~~~
\mathcal{L}_{0,2}^{\Lambda_0,\Lambda_0}\left((z_1,p),(z_2,-p)\right)=0~.
\end{eqnarray}
\textbf{Renormalization conditions at $\Lambda=0$}\\
The renormalization conditions are fixed at $\Lambda=0$ by imposing
\begin{eqnarray*}
   \fl a_l^{0,\Lambda_0}(z_1),~~s_l^{0,\Lambda_0}(z_1),~~
    d_l^{0,\Lambda_0}(z_1),~~b_l^{0,\Lambda_0}(z_1),~~
     c_l^{0,\Lambda_0}(z_1)
\end{eqnarray*}
to be smooth functions in $\mathcal{C}^{\infty}(\mathbb{R}^+)$, uniformly bounded w.r.t. $\Lambda_0$. \\
Typically all the renormalization conditions are assumed 
to be cutoff-independent. The simplest renormalization 
conditions are BPHZ-renormalization conditions, where we set 
\begin{eqnarray}\label{renoc}
    \fl \indent a_l^{0,\Lambda_0}(z_1)\equiv0,~~s_l^{0,\Lambda_0}(z_1)\equiv0,~~
    d_l^{0,\Lambda_0}(z_1)\equiv0,~~b_l^{0,\Lambda_0}(z_1)\equiv0,~~
     c_l^{0,\Lambda_0}(z_1)\equiv0\ .
\end{eqnarray}
These will be adopted in the following.
\section{Proof of renormalizability}
We define for all $n\geq 2$ and $0\leq r\leq 3$
\begin{eqnarray}\label{sh1}
    \fl \mathcal{L}_{l,n}^{\Lambda,\Lambda_0}(z_1;\vec{p}_n;
\left(z_1-z_i\right)^r\Phi_s)\nonumber\\
    :=\int_0^{\infty}dz_{2,n}(z_1-z_i)^r 
\mathcal{L}_{l,n}^{\Lambda,\Lambda_0}\left((z_1,p_1),\cdots,(z_n,p_n)\right)
\Phi_s(z_{2,s})~,
\end{eqnarray}
and for all $n\geq 3$
\begin{equation}\label{sh2}
    \fl  F_{12}\mathcal{L}_{l,n}^{\Lambda,\Lambda_0}(z_1,z_2;\vec{p}_n;\Phi_{s-1}):=(z_1-z_2)^3\int_0^{\infty}dz_{3,n}~\mathcal{L}_{l,n}^{\Lambda,\Lambda_0}\left((z_1,p_1),\cdots,(z_n,p_n)\right)\Phi_{s-1}(z_{3,s})~.
\end{equation}
For $n=2$ we define
\begin{equation}\label{weird}
    F_{12}\mathcal{L}_{l,2}^{\Lambda,\Lambda_0}(z_1,z_2;p):=(z_1-z_2)^3 \mathcal{L}_{l,2}^{\Lambda,\Lambda_0}\left((z_1,p),(z_2,-p)\right).
\end{equation}
\begin{theorem}(Boundedness)\label{thm}
We consider $0\leq\Lambda\leq \Lambda_0<\infty$, $1\leq s\leq n$, $2\leq i\leq n$, $2\leq j \leq s$ and $0\leq r\leq 3$. We consider test functions either of the form $\phi_{\tau_{2,s},y_{2,s}}(z_{2,n})$ or $\phi_{\tau_{2,s},y_{2,s}}^{(j)}(z_{2,n})$, which are also denoted in shorthand as $\phi_{\tau,y_{2,s}}$ resp. 
$\phi_{\tau,y_{2,s}}^{(j)}\,$, see (\ref{phit})-(\ref{phij}).
 Adopting (\ref{renoc})
 we claim
\begin{eqnarray}\label{c1}
    \fl(A)~\left| \partial^w \mathcal{L}_{l,n}^{\Lambda,\Lambda_0}(z_1;\vec{p}_n;\left(z_1-z_i\right)^r\phi_{\tau,y_{2,s}})\right|\nonumber\\\fl \indent~~~~~~\leq \left(\Lambda+m\right)^{4-n-|w|-r}\mathcal{P}_1\left(\log \frac{\Lambda+m}{m}\right)\mathcal{P}_2\left(\frac{\left\|\vec{p}_n\right\|}{\Lambda+m}\right) \mathcal{Q}_1\left(\frac{\tau^{-\frac{1}{2}}}{\Lambda+m}\right)\mathcal{F}^{\Lambda}_{s,l}(\tau)\ .
\end{eqnarray}
\begin{eqnarray}\label{c2}
    \fl(B)\left| F_{12} \mathcal{L}_{l,n}^{\Lambda,\Lambda_0}(z_1,z_2;\vec{p}_n;\phi_{\tau,y_{2,s}})\right|\leq \left(\Lambda+m\right)^{1-n}\mathcal{P}_3\left(\log \frac{\Lambda+m}{m}\right)\mathcal{P}_4\left(\frac{\left\|\vec{p}_n\right\|}{\Lambda+m}\right) \mathcal{F}^{12}_{s,l}(\Lambda,\tau)\ .
\end{eqnarray}
\begin{eqnarray}\label{c3}
    \fl (C)\left |\partial^w\mathcal{L}_{l,n}^{\Lambda,\Lambda_0}(z_1;\vec{p}_n;\phi_{\tau,y_{2,s}}^{(j)})\right|\nonumber\\\fl \indent ~~~~~~~\leq \left(\Lambda+m\right)^{3-n-|w|}\tau^{-\frac{1}{2}}_j\mathcal{P}_5\left(\log \frac{\Lambda+m}{m}\right)\mathcal{P}_6\left(\frac{\left\|\vec{p}_n\right\|}{\Lambda+m}\right)\mathcal{Q}_2\left(\frac{\tau^{-{\frac{1}{2}}}}{\Lambda+m}\right)\mathcal{F}^{\Lambda}_{s,l}(\tau)\ .
\end{eqnarray}
$(D)$ The functions defined in 
(\ref{sh1}), (\ref{sh2}) and (\ref{weird}) and their momentum derivatives are in
$\,\mathcal{C}^{\infty}\left(\mathbb{R}^+\right)\,$ w.r.t. $z_1\,$.\\
Here and in the following the $\mathcal{P}_i$ and $\mathcal{Q}_i$ denote (each time they appear possibly new) polynomials with nonnegative coefficients. The polynomials $\mathcal{Q}_i$ are reduced to a constant if $s=1$. The coefficients depend on $l,n,|w|,\delta$ but not on $\left \{ p_i \right \}$, $\Lambda$, $\Lambda_0$ and $z_1$. For $l=0$, all polynomials $\mathcal{P}_i$ reduce to constants. In the definition of $\mathcal{F}$ (\ref{treeStr}), $\delta>0$ may be chosen arbitrarily small.
\end{theorem}
\begin{remarks}
\indent \begin{itemize}
    \item [-]
The bounds (\ref{c2}) and (\ref{c3}) are required to close the inductive argument in the subsequent proof. The bound (\ref{c1}) is the central result of the boundedness Theorem \ref{thm} needed later to prove the convergence of $\partial^w \mathcal{L}_{l,n}^{\Lambda,\Lambda_0}(z_1;p_1,\cdots,p_n;\left(z_1-z_i\right)^r\phi_{\tau,y_{2,s}})$ in the limits $\Lambda\rightarrow 0$ and $\Lambda_0\rightarrow \infty$.
\item [-] The proof that we give in the following holds also for
larger classes of test functions indexed by a strictly positive parameter $\tau\,$ such that
$$\left|\partial_z^{\alpha}\phi_{\tau}(z)\right|\leq \tau^{-\frac{\alpha}{2}}\left|\phi_{\tau}(z)\right|~~~~\forall z\in \mathbb{R}^+\ ,~~~\forall\alpha\geq 0~.$$
The role of the parameter $\tau$ as it appears in the proof, is to absorb negative powers of the flow parameter $\Lambda$ by producing powers of $\frac{\tau}{\Lambda+m}$ that contribute to the polynomial $\mathcal{Q}$ at each step of the induction. This preserves the power counting in terms of $\Lambda+m$. We choose a simple example of these functions, which are the Robin heat kernels. They can be proved to be dense in $\mathcal{S}\left(\mathbb{R}^{+n}\right)$.
\item [-] The value of the integral 
\begin{equation}\label{deltaw}
    \int_0^{\infty} dw~\delta_w
\end{equation}
admits two possible choices, which are $1$ and $\frac{1}{2}$. These two choices are called respectively, the weak and strong definitions of the Dirac distribution \cite{8}.\\
The subsequent proof uses the strong definition of the Dirac distribution. For the weak definition, all the points of Theorem 1 hold except for (D). In the weak convention, the functions defined  in 
(\ref{sh1}), (\ref{sh2}) and (\ref{weird}) are in
$\,\mathcal{C}^{\infty}\left(\mathbb{R}^{+*}\right)\,$ w.r.t. $z_1\,$ and are not continuous at 0. One can verify that the proof of renormalizability is independent of the chosen convention. This comes from the fact that for a continuous function $f$, in both conventions, we have that
\begin{eqnarray*}
\int_0^{\infty}dz~\left(\int_0^{\infty}dz'~\delta(z-z')\right)f(z)=\int_0^{\infty}dz~f(z).
\end{eqnarray*}
\end{itemize}
\end{remarks}
\begin{proof}
The bounds are proven inductively using the standard inductive scheme which proceeds upwards in $l$, 
for given $l$ upwards in $n$, and for given $(n,l)$ downwards in $|w|$ starting from some arbitrary $|w_{\max}|\geq 3$. 
The induction works because the terms on the r.h.s. of the FE always are prior to the one of the l.h.s. in the inductive order. So the bounds (\ref{c1})-(\ref{c3}) may be used as an induction hypothesis on the r.h.s. Once verified in the first induction step, we integrate the FE,
where the terms with $n+|w|+r\geq 5$ are integrated down from $\Lambda_0$ to $\Lambda$ because of the boundary conditions (\ref{ii}), and the  terms with $n+|w|+r\leq 4$ at the renormalization point are integrated upwards from $0$ to $\Lambda$ using (\ref{renoc}). We can write remembering (\ref{13*})
\begin{eqnarray}\label{renopoint}
\fl{\partial}^{w}\mathcal{L}_{l,n}^{\Lambda,\Lambda_0}\left(z_1;\vec{0};\phi_{\tau,y_{2,s}}\right)\nonumber\\={\partial}^{w}\mathcal{L}_{l,n}^{0,\Lambda_0}\left(z_1;\vec{0};\phi_{\tau,y_{2,s}}\right)+\int_{0}^{\Lambda} 
d\lambda\ \partial_{\lambda}\ {\partial}^{w}\mathcal{L}_{l,n}^{\lambda,\Lambda_0}\left(z_1;\vec{0};\phi_{\tau,y_{2,s}}\right)\ .
\end{eqnarray}
Once a bound has been obtained at the renormalization point, it is possible to move away from the renormalization point using the integrated Taylor formula
\begin{eqnarray}
\label{taylorfor}
\fl{\partial}^{w}\mathcal{L}_{l,n}^{\Lambda,\Lambda_0}\left(z_1;\vec{p}_n;\phi_{\tau,y_{2,s}}\right)={\partial}^{w}\mathcal{L}_{l,n}^{\Lambda,\Lambda_0}\left(z_1;\vec{0};\phi_{\tau,y_{2,s}}\right)\nonumber\\+\sum_{i=1}^n \sum_{\mu=1}^4 {p}_{i,\mu} \int_0^1 dt \left( {\partial}_{{p}_{i,\mu}} {\partial}^{w}\mathcal{L}_{l,n}^{\Lambda,\Lambda_0}\right)\left(z_1;t\vec{p}_n;\phi_{\tau,y_{2,s}}\right).\qquad
\end{eqnarray}
The induction starts with the pair $(0,4)$ for which the r.h.s. of the FE vanishes so that 
$$\mathcal{L}_{0,4}^{\Lambda,\Lambda_0}\left((z_1,p_1),\cdots,(z_4,p_4)\right)=\lambda\prod_{i=2}^4 \delta(z_1-z_i)$$
which implies for the test function $\phi_{\tau,y_{2,s}}$ defined in (\ref{phit}) that
\begin{equation}\label{4pt}
    \mathcal{L}_{0,4}^{\Lambda,\Lambda_0}\left(z_1;\vec{p}_4;\phi_{\tau,y_{2,s}}\right)=\lambda\prod_{i=2}^s p_R(\tau_i;z_1,y_i)~.
\end{equation}
Using the bounds (\ref{prb}) and (\ref{pbdelta}) from Appendix A, we deduce that $$\prod_{i=2}^s p_R(\tau_i;z_1,y_i)\leq 2^s \prod_{i=2}^s p_B\left(\tau_i;z_1,y_i\right)\leq 2^s(1+\delta)^{\frac{s}{2}}~\prod_{i=2}^s p_B\left(\tau_{i,\delta};z_1,y_i\right),$$ 
which implies that that (\ref{4pt}) can be bounded by a tree  with no internal vertices and with a root vertex $z_1$ linked to the external vertices $y_{2,s}\,$, 
which is in agreement with the bound (\ref{c1}). The constants are absorbed in the polynomial $\mathcal{P}_1$, which is of degree $0$ at the tree order.\\
\underline{(I) Bounds on the r.h.s. of the FE:}\\
We want to establish the bounds
\begin{eqnarray}\label{borne}
\fl|\partial_{\Lambda}\partial^{w}
\mathcal{L}^{\Lambda,\Lambda_0}_{l,n}
\left(z_1;\vec{p}_n;\left(z_1-z_i\right)^r\phi_{\tau,y_{2,s}}\right)|
\leq \left(\Lambda+m\right)^{3-n-|w|-r}\mathcal{P}_1\left(\log \frac{\Lambda+m}{m}\right)
   \mathcal{P}_2\left(\frac{\left\|\vec{p}_n\right\|}{\Lambda+m}\right)\nonumber\\
   \times\mathcal{Q}_1\left(\frac{\tau^{-\frac{1}{2}}}{\Lambda+m}\right)\mathcal{F}^{\Lambda}_{s,l}\left(\tau\right)\ .
\end{eqnarray}
(A) In this part we consider the case $r=0$.\\
($A_1$) Let $R_1$ be the first term on the RHS of the FE
$$R_1:=\int_p\int_{z_{2,n},z,z'}\partial^w \mathcal{L}_{n+2,l-1}^{\Lambda,\Lambda_0}\left((\vec{z_n},\vec{p_n}),(z,p),(z',-p)\right)\dot{C}_{R}^{\Lambda}
\left(p;z,z'\right)\prod_{i=2}^sp_R(\tau_i;z_i,y_i)$$
which can be written as 
\begin{equation}\label{R1*}
\fl -\frac{2e^{-\frac{p^2+m^2}{\Lambda^2}}}{\Lambda^3}\int_p\int_0^{\infty}du~\partial^w \mathcal{L}_{n+2,l-1}^{\Lambda,\Lambda_0}\left(z_1;p,-p,\vec{p_n};\phi_{\tau,y_{2,s}}\times p_{R}(\frac{1}{2\Lambda^2};u,\cdot)\times  
p_{R}(\frac{1}{2\Lambda^2};\cdot,u)\right)\ ,
\end{equation}
where we used (\ref{p3}) and the semi-group property for $p_R$, see (\ref{sgr}).\\
Applying the induction hypothesis gives the bound 
\begin{eqnarray}\label{79}
   \fl|R_{1}|\leq\left(\Lambda+m\right)^{2-n-|w|} \mathcal{P}_1\left(\log \frac{\Lambda+m}{m}\right)\mathcal{P}_2\left(\frac{\left\|\vec{p}_n\right\|}{\Lambda+m}\right)\mathcal{Q}_1\left(\frac{\tau^{-\frac{1}{2}}}{\Lambda+m}\right)\nonumber\\
   \times \int_0^{\infty}du\int_{\vec{z}}\sum_{T^{s+2}_{l-1}(z_1,y_{2,s},u,u)}\mathcal{F}\left(\Lambda,\left\{\tau,\frac{1}{2\Lambda^2},\frac{1}{2\Lambda^2}\right\};T^{s+2}_{l-1}(z_1,y_{2,s},u,u,\vec{z})\right).
\end{eqnarray}
In the sequel, we write the polynomials $\mathcal{P}_1$, $\mathcal{P}_2$ and $\mathcal{Q}_1$ without their lower indices. One should keep in mind that these polynomials may have, each time they appear,
 different positive coefficients which depend on $l,n,|w|,\delta$ only and not on $\left\{p_i\right\}$, $\Lambda$, $\Lambda_0$ and $z_1$.\\
For any contribution to (\ref{79}) we denote by $z', z''$ the vertices in the tree $T^{s+2}_{l-1}(z_1,y_{2,s},u,u)$ to which the test functions $p_B(\frac{1+\delta}{2\Lambda^2};u,\cdot)$ and $p_B(\frac{1+\delta}{2\Lambda^2};\cdot,u)$ are attached. Interchanging $\int_{\vec{z}}$ and $\int_u$ and performing the integral over $u$ using the semi-group property (\ref{sgr}), we obtain
$$\int_0^{\infty}du~ p_B(\frac{1+\delta}{2\Lambda^2};z',u)\ 
p_B(\frac{1+\delta}{2\Lambda^2};u,z'')\leq p_B(\frac{1+\delta}{\Lambda^2};z',z'')\leq O(1) \Lambda$$
with a positive constant $O(1)\,$.
Using this bound we obtain $$\int_0^{\infty}du~\mathcal{F}\left(\Lambda,\left\{\tau,\frac{1}{2\Lambda^2},\frac{1}{2\Lambda^2}\right\};T^{s+2}_{l-1};z_1,y_{2,s},u,u\right)\leq O(1) \Lambda~ \mathcal{F}\left(\Lambda,\tau;T^{s}_{l};z_1,y_{2,s}\right),$$
where the tree $T^{s}_l$ is the reduced tree obtained from $T^{s+2}_{l-1}$ by taking away the two external lines ending in $u$. Note that $v_2$, the number of vertices of incidence number $2$, verifies $v_2+\delta_{c_1,1}\leq 3l-3-2+s/2+1\leq 3l-2+s/2$, which explains the reduction of $T^{s+2}_{l-1}$ is in $\mathcal{T}^s_l$. The elimination of vertices of incidence number $1$ together with their adjacent line is justified by the fact that $\int_{z'}p_B(\frac{1+\delta}{\Lambda_I^2};z',z'')\leq 1$. The reduction process for each tree fixes uniquely the set of internal vertices of $T^s_l$ in terms of those of $T^{s+2}_{l-1}$. Therefore, we obtain 
\begin{eqnarray}\label{bR1b}
\fl|R_{1}|\leq\left(\Lambda+m\right)^{3-n-|w|} \mathcal{P}\left(\log \frac{\Lambda+m}{m}\right) \mathcal{P}\left(\frac{\left\|\vec{p}_n\right\|}{\Lambda+m}\right)\mathcal{Q}\left(\frac{\tau^{-\frac{1}{2}}}{\Lambda+m}\right)\nonumber\\
   \times  \sum_{T^{s}_{l}(z_1,y_{2,s})}\mathcal{F}\left(\Lambda,\tau;T^{s}_{l};z_1,y_{2,s}\right).
\end{eqnarray}
($A_2$) We now consider the second term on the RHS of the FE. It is enough to analyse the following term from the symmetrized sum 
 \begin{eqnarray*}
      \fl R_2:=\int_{z_{2,n},z,z'}\partial^{w_1}\mathcal{L}^{\Lambda,\Lambda_0}_{l_1,n_1+1}((z_1,p_1),\cdots,(z_{n_1},p_{n_1}),(z,p))\partial^{w_3}\dot{C}_R^{\Lambda}(p;z,z')\\ \times\partial^{w_2}\mathcal{L}^{\Lambda,\Lambda_0}_{l_2,n_2+1}((z_{n_1+1},p_{n_1+1}),\cdots,(z_{n},p_{n}),(z',-p))\phi_{\tau,y_{2,s}}(z_{2,n})\ ,
 \end{eqnarray*}
 in which the arguments $(z_i,p_i)$ appear in  $\mathcal{L}^{\Lambda,\Lambda_0}_{l_1,n_1+1}$ and $\mathcal{L}^{\Lambda,\Lambda_0}_{l_2,n_2+1}$ in an ordered way. \\ $R_2$ can be rewritten similarly as in ($A_1$)
\begin{eqnarray*}
      \fl R_2:=\int_u\int_{z_{2,n},z,z'}\partial^{w_1}\mathcal{L}^{\Lambda,\Lambda_0}_{l_1,n_1+1}((z_1,p_1),\cdots,(z_{n_1},p_{n_1}),(z,p))\partial^{w_3}\dot{C}^{\Lambda}(p)\\ \fl \times \ \partial^{w_2}\mathcal{L}^{\Lambda,\Lambda_0}_{l_2,n_2+1}((z_{n_1+1},p_{n_1+1}),\cdots,(z_{n},p_{n}),(z',-p)) \phi_{\tau,y_{2,s}}(z_{2,n})\\ 
\times\ p_R(\frac{1}{2\Lambda^2};z,u)\ p_R(\frac{1}{2\Lambda^2};u,z')\ .
 \end{eqnarray*}
We define
 $$\phi'_{s_1}(z_{2,n_1})=\prod^{s_1}_{r=2}p_R(\tau_r;z_r,y_r)
,~~~\phi''_{s_2}(z_{n_1+1,n-1})
=\prod_{r=n_1+1}^{s_2+n_1}\! p_R(\tau_r;z_r,y_r)~,$$
where $s=s_1+s_2$.\\
 Therefore $\,R_2$ can be rewritten as 
 \begin{eqnarray}\label{R2}
     \fl R_2=\int_{z_n}\int_u  \partial^{w_1}
\mathcal{L}^{\Lambda,\Lambda_0}_{l_1,n_1+1}
\Bigl(z_1;p_1,\cdots,p_{n_1},p;\phi'_{s_1}\times p_R(\frac{1}{2\Lambda^2};.,u)
\Bigr)
\, \partial^{w_3}\dot{C}^{\Lambda}(p)\\
     \times \partial^{w_2}\mathcal{L}^{\Lambda,\Lambda_0}_{l_2,n_2+1}
\Bigl(z_n;-p,p_{n_1+1},
\cdots,p_n;\phi''_{s_2}\times p_R(\frac{1}{2\Lambda^2};u,.)\Bigr)
\phi_n(z_n)\ ,\nonumber
 \end{eqnarray}
 where $\phi_n(z_n)=p(\tau_n;z_n,y_n)$ if $s=n$, and $\phi_n(z_n)\equiv1$ otherwise.
 Applying the induction hypothesis to both terms in (\ref{R2}) 
we obtain the bound
 \begin{eqnarray}\label{84}
     \fl|R_2|\leq \left(\Lambda+m\right)^{8-n-|w|-2-3}\mathcal{P}\left(\log \frac{\Lambda+m}{m}\right)\mathcal{P}\left(\frac{\left\|\vec{p}_n\right\|}{\Lambda+m}\right)\mathcal{Q}\left(\frac{\tau^{-\frac{1}{2}}}{\Lambda+m}\right)\nonumber\\\indent \indent\times \int_{z_n}\int_u \sum_{T^{s_1+1}_{l_1},T^{s_2+1}_{l_2}}\mathcal{F}\left(\Lambda,\left\{\tau',\frac{1}{2\Lambda^2}\right\};T^{s_1+1}_{l_1};z_1,y_{2,s_1},u\right)\nonumber\\\indent\indent \indent\times \mathcal{F}\left(\Lambda,\left\{\tau'',\frac{1}{2\Lambda^2}\right\};T^{s_2+1}_{l_2};z_n,u,y_{s_1+1},\cdots,y_{s(n)}\right)\phi_n(z_n)~,
\end{eqnarray}
 where we used 
\begin{equation}\label{propagateur}
|\partial^{w_3}\dot{C}^{\Lambda}(p)|\leq 
\left(\Lambda+m\right)^{-3-|w_3|}\mathcal{P}\left(\frac{|p|}{\Lambda+m}\right).
\end{equation}
 In (\ref{84}) we set $s(n)=s$ if $s<n$ and $s(n)=s-1$ if $s=n\,$. \\
 Interchanging the integral over $u$ with the sum over trees we obtain
 \begin{eqnarray*}
 |R_2|\leq \left(\Lambda+m\right)^{3- n-|w|}\mathcal{P}\left(\log \frac{\Lambda+m}{m}\right)\mathcal{P}\left(\frac{\left\|\vec{p}_n\right\|}{\Lambda+m}\right)\mathcal{Q}\left(\frac{\tau^{-\frac{1}{2}}}{\Lambda+m}\right) \\\indent \indent \times \sum_{T^s_l(T^{s_1+1}_{l_1},T^{s_2+1}_{l_2})}\int_{z_n}\mathcal{F}\left(\Lambda,\tau;T^{s}_{l};z_1,y_{2,s}\right),
 \end{eqnarray*}
 with the following explanations:
 \begin{itemize}
     \item Any contribution in the sum over trees $T^{s}_l\left(T^{s_1+1}_{l_1},T^{s_2+1}_{l_2}\right)(z_1,y_2,\cdots,y_s,\vec{z})$ is obtained from $T^{s_1+1}_{l_1}(z_1,y_2,\cdots,y_{s_1},u,\vec{z}^{~'})$ and $T^{s_2+1}_{l_2}(z_n,y_2,\cdots,y_{s_1+1},u,\vec{z}^{~''})$ by joining these two trees via the lines going from the vertices $z'$ and $z''$  to $u$, where $z'$ and $z''$ are the vertices attached to $u$ in the two trees. These two lines have parameters $\frac{1+\delta}{2\Lambda^2}$. We use (\ref{sgr}) from Appendix A to obtain
     $$\int_u p_B(\frac{1+\delta}{2\Lambda^2};z',u)\ 
p_B(\frac{1+\delta}{2\Lambda^2};u,z'')\leq p_B(\frac{1+\delta}{\Lambda^2};z',z'')$$
so that the new internal line has a parameter in the interval $\left[\Lambda,\Lambda_0\right]$ over which the sup is taken in the definition of $\mathcal{F}$.
\item When performing the integral over $z_n$, we note that $z_n$ has been viewed as the root vertex of $T^{s_2+1}_{l_2}(z_n,u,y_{s_1+1},\cdots,y_s,\vec{z})$. We distinguish two cases:
\begin{itemize}
    \item If $s=n$, we set $\phi_n(z_n)=p_B(\tau_n;z_n,y_n)$, and $z_n$ becomes an internal vertex and $y_n$ an external vertex of $T^s_l$.
    \item If $s<n$, then $\phi_n(z_n)\equiv 1$ and the vertex $z_n$ becomes an internal vertex of $T^s_l$ unless $c(z_n)\equiv 1$. In this case, we use 
    \begin{equation}\label{zk}
    \int_{z_n}p_B(t_{n,\delta};z_n,z_j)\leq1~.
    \end{equation}
    This integration permits to take away the vertex $z_n$ and the internal line joining it to an internal vertex $z_j$ of the tree $T^{s_2+1}_{l_2}$. If $c(z_j)=2$, let $z_l$ be the other vertex to which $z_j$ is joined. The semi-group property (\ref{sgr}) implies
    $$\int_{z_j}p_B(t_{n,\delta};z_{n},z_{j})~p_B(t_{j,\delta};z_j,z_l)\leq p_B(t_{n,\delta}+t_{j,\delta};z_n,z_l)\ .$$
    This elimination process continues until we reach $c(z_k)>2$. Applying (\ref{zk}) for $z_j=z_k$ takes away the internal vertex $z_n$. If $c(z_k)=3$ in $T^{s_2+1}_{l_2}$ then it becomes equal to $2$ after integrating out $z_n$. Therefore
 a new vertex of incidence number $2$ is created in the new tree, and this is compatible with the definition of $T^s_l$. Namely, let $v_{2,i}$ be the number of vertices with incidence number $2$ of the tree $T^{s_i+1}_{l_i}$. By definition of $T^{s_i+1}_{l_i}$, we have 
    \begin{equation*}
        v_{2,1}+\delta_{c_1,1}\leq 3l_1-2+\frac{s_1+1}{2},~~~v_{2,2}+\delta_{c_s,1}\leq 3l_2-2+\frac{s_2+1}{2}\ .
    \end{equation*}
    Since $c_s:=c(z_s)=1$ we deduce that 
    $$v_{2,2}\leq 3l_2-3+\frac{s_2+1}{2}\ .$$
    The number of vertices $v_2$ with incidence number $2$ in the new tree obtained from $T^s_l\left(T^{s_1+1}_{l_1},T^{s_2+1}_{l_2}\right)$ after integrating out $z_n$ is equal to $v_{2,1}+v_{2,2}+1$ so that 
    \begin{equation*}
        v_{2}+\delta_{c_1,1}\leq 3l-3+\frac{s}{2}~.
    \end{equation*}
\end{itemize}
Therefore we conclude that
$$|R_2|\leq 
\left(\Lambda+m\right)^{3-n-|w|}
\mathcal{P}\!\left(\log \frac{\Lambda+m}{m}\right)
\mathcal{P}\!\left(\frac{\left\|\vec{p}_n\right\|}{\Lambda+m}\right)
\mathcal{Q}\!
\left(\frac{\tau^{-\frac{1}{2}}}{\Lambda+m}\right)\!\sum_{T^{s}_{l}(z_1,y_{2,s})}\!
\mathcal{F}\!\left(\Lambda,\tau;T^{s}_{l};z_1,y_{2,s}\right)~.$$
 \end{itemize}
(B) We consider now the case $r\neq 0$:\\
For the first term on the r.h.s. of the flow equation (\ref{FEL}) the bounds are proven exactly as in ($A_1$). For the second term we proceed similarly as in ($A_2$). We pick a generic term on the r.h.s.,
\begin{eqnarray*}
   \fl \int_u\int_{z_{2,n},z,z'}\phi_{\tau,y_{2,s}}(z_{2,n})\partial^{w_3}
\,\dot{C}^{\Lambda}(p)\ p_R(\frac{1}{2\Lambda^2};z,u)\ 
p_R(\frac{1}{2\Lambda^2};u,z')(z_i-z_1)^r\\ \fl\times\partial^{w_1}
\mathcal{L}^{\Lambda,\Lambda_0}_{l_1,n_1+1}\left((z_1,p_1),\cdots,(z_{n_1},p_{n_1}),(z,p)\right)\ 
\partial^{w_2}\mathcal{L}^{\Lambda,\Lambda_0}_{l_2,n_2+1}\left((z_{n_1+1},p_{n_1+1}),
\cdots,(z_{n},p_{n}),(z',-p)\right)~.
\end{eqnarray*}
In the case where $i\leq n_1$ the proof is the same as for $r=0$, 
up to inserting the modified induction hypothesis for 
\begin{eqnarray*}
\fl\partial^{w_1}\mathcal{L}_{l_1,n_1+1}^{\Lambda,\Lambda_0}
\left(z_1;\vec{p}_{n};(z_i-z_1)^r\phi'_{s_1}\times p_R(\frac{1}{2\Lambda^2};
\cdot,u)\right)\\
\fl\indent
=\int_{z_2,\cdots,z_{n_1},z}\partial^{w_1}
\mathcal{L}^{\Lambda,\Lambda_0}_{l_1,n_1+1}
((z_1,p_1),\cdots,(z_{n_1},p_{n_1}),(z,p))(z_i-z_1)^r
\phi'_{s_1}(z_{2,n_1})p_R(\frac{1}{2\Lambda^2};z,u)~.
\end{eqnarray*}
If $i>n_1$ we assume without restriction $i=n$ and proceed 
again as in ($A_2$) to obtain the bound 
\begin{eqnarray*}
\fl (\Lambda+m)^{3-n-|w|}\mathcal{P}\left(\log \frac{\Lambda+m}{m}\right)
   \mathcal{P}\left(\frac{\left\|\vec{p}_n\right\|}{\Lambda+m}\right)
\mathcal{Q}\left(\frac{\tau^{-\frac{1}{2}}}{\Lambda+m}\right)
    \\ \times \int_{z_n}\int_u |z_n-z_1|^r \sum_{T^{s_1+1}_{l_1},T^{s_2+1}_{l_2}}
\mathcal{F}\left(\Lambda,\left\{\tau',\frac{1}{2\Lambda^2}\right\};
T^{s_1+1}_{l_1};z_1,y_{2,s_1},u\right)\\
     \indent \times \mathcal{F}\left(\Lambda,\left\{\tau'',
\frac{1}{2\Lambda^2}\right\};
T^{s_2+1}_{l_2};z_n,u,y_{s_1+1},\cdots,y_{s(n)}\right)\phi_n(z_n)~.
\end{eqnarray*}
We bound
\begin{equation}\label{dist}
    |z_n-z_1|\leq \sum_{a=1}^q |v_a-v_{a-1}|~,
\end{equation}
where $\left\{v_a\right\}$ are the positions of the internal vertices in the tree $T^s _l(T^{s_1+1}_{l_1+1},T^{s_2+1}_{l_2+1})$ defined as in ($A_2$), 
on the path joining $z_1=v_0$ and $z_n=v_q\,$. Using the inequality (\ref{in1}) from the Appendix A for $\tau=\frac{1+\delta}{\Lambda_a^2}$, we obtain (\ref{borne}). 
Note that the cases $s=n$ and $s<n$ are treated as in $(A_2)$.\\
The previous reasoning holds also for 
$\partial_{\Lambda}F_{12}
\mathcal{L}_{l,n}^{\Lambda,\Lambda_0}(z_1,z_2;\vec{p}_n;\phi_{\tau,y_{2,s}})$, 
where $z_2$ takes the role of $z_n$. 
After absorbing all constants in $\mathcal{P}$, we obtain 
\begin{eqnarray}\label{lam0}
    \fl\left| \partial_{\Lambda}\partial^w \mathcal{L}_{l,n}^{\Lambda,\Lambda_0}(z_1;\vec{p}_n;(z_1-z_i)^r\phi_{\tau,y_{2,s}})\right|\leq \left(\Lambda+m\right)^{3-n-|w|-r}\mathcal{P}\left(\log \frac{\Lambda+m}{m}\right)\mathcal{P}\left(\frac{\left\|\vec{p}_n\right\|}{\Lambda+m}\right)\nonumber\\
\times\ \mathcal{Q}\left(\frac{\tau^{-\frac{1}{2}}}{\Lambda+m}\right)
\mathcal{F}^{\Lambda}_{s,l}(\tau)~,\label{brel1}\\
   \fl \left| \partial_{\Lambda}F_{12} \mathcal{L}_{l,n}^{\Lambda,\Lambda_0}(z_1,z_2;\vec{p}_{n};\phi_{\tau,y_{2,s}})\right|\leq \left(\Lambda+m\right)^{-n}\mathcal{P}\left(\frac{\left\|\vec{p}_n\right\|}{\Lambda+m}\right)\mathcal{P}\left(\log \frac{\Lambda+m}{m}\right) \mathcal{F}^{12}_{s,l}(\Lambda,\tau)~.\label{12}
\end{eqnarray}
\noindent The bounds for (\ref{rel})-(\ref{lolo})
\begin{eqnarray}
    \fl|\partial_{\Lambda}c_l^{\Lambda,\Lambda_0}(z_1)|\leq \left(\Lambda+m\right)^{-1}\mathcal{P}\left(\log \frac{\Lambda+m}{m}\right),~~|\partial_{\Lambda}a_l^{\Lambda,\Lambda_0}(z_1)|\leq \left(\Lambda+m\right)\mathcal{P}\left(\log \frac{\Lambda+m}{m}\right),\label{brel}\\
    \fl |\partial_{\Lambda}b_l^{\Lambda,\Lambda_0}(z_1)|\leq \left(\Lambda+m\right)^{-1}\mathcal{P}\left(\log \frac{\Lambda+m}{m}\right),~~|\partial_{\Lambda}d_l^{\Lambda,\Lambda_0}(z_1)|\leq \left(\Lambda+m\right)^{-1}\mathcal{P}\left(\log \frac{\Lambda+m}{m}\right),\\
    |\partial_{\Lambda}s_l^{\Lambda,\Lambda_0}(z_1)|\leq \mathcal{P}\left(\log \frac{\Lambda+m}{m}\right)\label{bre}
\end{eqnarray}
are obtained on restricting the previous considerations to the case $s=1$, in which all the coordinates $z_{2,n}$ are integrated over with $n=2$ or $n=4$ and the momenta $\vec{p}_n$ set to $\vec{0}$. \\
(C) We come to the bound on $\partial_{\Lambda}\partial^w\mathcal{L}_{l,n}^{\Lambda,\Lambda_0}(z_1;\vec{p}_{n};\phi_{\tau,y_{2,s}}^{(j)})$ (i.e. remember (\ref{phij}-\ref{phij'})) . As compared to ($A$), the only case which requires new analysis is the bound on the second term from the r.h.s. of the FE (\ref{FEL}) in the case $j>s_1$. Then we assume without restriction, similarly as in ($A$), that $j=s$. The term to be bounded corresponding to (\ref{84}) is then 
\begin{eqnarray}\label{98}
    \fl\left(\Lambda+m\right)^{3-n-|w|}\mathcal{P}\left(\log \frac{\Lambda+m}{m}\right)\mathcal{P}\left(\frac{\left\|\vec{p}_n\right\|}{\Lambda+m}\right)\mathcal{Q}\left(\frac{\tau^{-\frac{1}{2}}}{\Lambda+m}\right)\nonumber\\\fl\indent\times\int_u\sum_{T^{s_1+1}_{l_1},T^{s_2+1}_{l_2}}\mathcal{F}\left(\Lambda,\left\{\tau',\frac{1}{2\Lambda^2}\right\};T^{s_1+1}_{l_1};z_1,y_{2,s_1},u\right)\nonumber\\
     \fl \indent \indent  \times e^{-\frac{m^2}{2\Lambda^2}}\int_{z_s}  \mathcal{F}\left(\Lambda,\left\{\tau^{''},\frac{1}{2\Lambda^2}\right\};T^{s_2+1}_{l_2};z_s,u,y_{s_1+1},\cdots,y_{s(n)}\right)|p_R^{(1)}(\tau_s;z_s,z_1;y_s)|~.
\end{eqnarray}
We factorized $e^{-\frac{m^2}{2\Lambda^2}}$ from the derivative of the flowing propagator $\dot{C}^{\Lambda}(p)$, and we will use the bound
\begin{equation}\label{étoile}
\left(\Lambda+m\right)e^{-\frac{m^2}{2\Lambda^2}} \,\leq \,C\ \Lambda~,
~~~C:= \|(1+x)e^{-\frac{x^2}{2}}\|_{\infty}~.
\end{equation}
To bound (\ref{98}) we telescope the difference 
$\,p_R^{(1)}(\tau_s;z_s,z_1;y_s)$ along the tree 
$\,T^s_l(T^{s_1+1}_{l_1},T^{s_2+1}_{l_2})$ 
similarly as in (\ref{dist}). We then have to bound expressions of the type
\begin{equation}\label{stor}
p_B (\frac{1+\delta}{\Lambda_I^2};v_{a-1},v_a)\left
|\,p_R(\tau_s;v_a,y_s)-p_R(\tau_s;v_{a-1},y_s)\right|,
\end{equation}
where $v_{a-1}$, $v_a$ are adjacent internal vertices in $T^s_l(T^{s_1+1}_{l_1},T^{s_2+1}_{l_2})$ on the unique path from $z_1$ to $y_s$. 
Taylor expansion of $p_R(\tau_s;v_a,y_s)$ gives
\begin{equation*}
    p_R(\tau_s;v_a,y_s)=p_R(\tau_s;v_{a-1},y_s)
+\int_0^1 dt ~(\partial_t{p}_R)(\tau_s;tv_{a-1}+(1-t)v_a,y_s)\ .
\end{equation*}
Lemma \ref{lemma2} from Appendix A implies that for all $0<\delta'<1$, we have
\begin{equation*}
    \fl~~~~~~|p_R^{(1)}(\tau_s;v_a,v_{a-1},y_s)|\leq ~C_{1,\delta}~
\frac{|v_a-v_{a-1}|}{\sqrt{\tau_s}}
\int_0^1 dt~ {p}_B(\tau_{s,\delta'};tv_{a-1}+(1-t)v_a,y_s)~.
\end{equation*}
Therefore (\ref{stor}) is bounded by 
\begin{equation}\label{sisi}
     C_{1,\delta}\tau_s^{-\frac{1}{2}}~ |v_a-v_{a-1}|~
p_B(\frac{1+\delta}{\Lambda_I^2};v_{a-1},v_a)
\int_0^1 ~dt~{p}_B((1+\delta')\tau_{s};tv_{a-1}+(1-t)v_a,y_s)~.
\end{equation}
Introducing for $2\delta<1$,
\begin{eqnarray*}
    b=2~\frac{1+2\delta}{1-2\delta}, 
\end{eqnarray*}
we distinguish  between the two cases:
\begin{itemize}
    \item Case (1): $\delta'\Lambda^2\leq b \tau_s^{-1}$\\
    Using (\ref{étoile}), we obtain
    $$
    (\Lambda+m)^{3-n-|w|}e^{-\frac{m^2}{2\Lambda^2}}\leq C\left(\frac{b}{\delta'}\right)^{\frac{1}{2}} (\Lambda+m)^{2-n-|w|}\tau_s^{-\frac{1}{2}}~.$$
The tree $T^s_l(T^{s_1+1}_{l_1},T^{s_2+1}_{l_2})$ is obtained from the two initial trees by joining them via $u$ as in ($A_2$), and we bound 
\begin{equation}
|p_R^{(1)}(\tau_s;z_s,z_1;y_s)|\leq |p_R(\tau_s;z_s,y_s)|+|p_R(\tau_s;z_1,y_s)|~.
\end{equation}
Here $p_R(\tau_s,z_1,y_s)$ is associated to $\phi_s(z_s)\equiv1$, and the integration over 
$z_s$ in $T^s_l(T^{s_1+1}_{l_1},T^{s_2+1}_{l_2})$ 
is performed similarly as in ($A_2$). This implies that for $\delta'\Lambda^2\leq b \tau_s^{-1}$, (\ref{98}) is bounded by 
\begin{eqnarray}\label{bornne1}
     \fl~~~~~\left(\frac{b}{\delta'}\right)^{\frac{1}{2}} (\Lambda+m)^{2-n-|w|}\tau_s^{-\frac{1}{2}}\mathcal{P}\!\left(\log \frac{\Lambda+m}{m}\right)
\mathcal{P}\!\left(\frac{\left\|\vec{p}_n\right\|}{\Lambda+m}\right)
\mathcal{Q}\!\left(\frac{\tau^{-{\frac{1}{2}}}}{\Lambda+m}\right)
\mathcal{F}^{\Lambda}_{s,l}(\tau)~.
\end{eqnarray}
\item Case (2): $\delta'\Lambda^2\geq b\tau_s^{-1}$\\
Using Lemma \ref{lemme1} we obtain the following bound
\begin{eqnarray*}
\fl |v_a-v_{a-1}|~p_B(\frac{1+\delta}{\Lambda_I^2};v_{a-1},v_a)~\int_0^1 ~ dt~{p}_B(\tau_{s,\delta'};tv_{a-1}+(1-t)v_a,y_s)\\
\leq C_{\delta}~\Lambda^{-1} ~p_B(\frac{2}{\Lambda_I^2};v_{a-1},v_a)
p_B((1+\delta')^3{\tau}_s;v_{a-1},y_s)\ ,
\end{eqnarray*}
which implies that (\ref{sisi}) can be bounded by 
\begin{eqnarray*}
\fl ~~~~ C_{\delta}C_{1,\delta}~\Lambda^{-1} \tau_s^{-\frac{1}{2}}
p_B(\frac{2}{\Lambda_I^2};v_{a-1},v_a)
~p_B((1+\delta')^3{\tau}_s;v_{a-1},y_s)~. 
\end{eqnarray*}
Choosing $\delta'$ such that $(1+\delta')^3=1+\delta$, that is $\delta'=\frac{\delta}{3}+O(\delta^2)$ and using the bound (\ref{pbdelta}), the final bound obtained for (\ref{stor}) reads then
\begin{equation}\label{cuelo}
\fl ~~~~C'_{\delta} ~
\tau_s^{-\frac{1}{2}}\Lambda^{-1}\int_0^{\infty}dv
~p_B(\frac{1+\delta}{\Lambda_I^2};v_{a-1},v)
\ p_B(\frac{1+\delta}{\Lambda_I^2};v,v_a)\,
p_B({\tau}_{s,\delta};v_{a-1},y_s)~,
\end{equation}
where we used the property (\ref{rr+}) from Appendix A and $C'_{\delta}=C_{\delta}C_{1,\delta}$.
\end{itemize}
\noindent
Note that the addition of a new internal vertex $v$ of incidence number $2$ in (\ref{cuelo}) is compatible with the inequality $v_2+\delta_{c_1,1}\leq 3l-2+s/2$, since $v'_2$, the number of vertices with incidence number $2$ of the new tree is equal to $v_{2,1}+v_{2,2}+1$, where $v_{2,i}$ is the number of internal vertices of incidence number $2$ of the tree $T^{s_i+1}_{l_i}$.
Therefore 
$$v'_2+\delta_{c_1,1}=v_{2,1}+v_{2,2}+1+\delta_{c_1,1}\leq 3(l_1+l_2)-4+\frac{s_1+s_2+2}{2}+1=3l-2+s/2\ .$$
The Case (2) corresponds to a sum of two new trees of type $\mathcal{T}_l^s\,$, 
where in comparison to $T_l^s\left(T_{l_1}^{s_1+1},T_{l_2}^{s_2+1}\right)$, the incidence number of $v_{a-1}$ has increased by one unit. In (\ref{98}) an integral over $z_s$ is performed. If in the new tree \\
a) $z_s$ has $c(z_s)>1$, then $z_s$ takes the role of an internal vertex of the new tree,\\
b) $z_s$ has $c(z_s)=1$ we integrate over $z_s$ using (\ref{intc}) so that the vertex $z_s$ disappears. 
As a consequence of the bounds (\ref{bornne1}) and (\ref{cuelo}), on replacing again $s\rightarrow j$  we thus obtain for $n\geq 2$
\begin{eqnarray*}
   \fl \left|\partial_{\Lambda}
\partial^w\mathcal{L}_{l,n}^{\Lambda,\Lambda_0}(z_1;p_{1,n};\phi_{\tau,y_{2,s}}^{(j)})\right|
\leq \left(\left(\frac{b}{\delta}\right)^{\frac{1}{2}}+C'_{\delta}\right)\left(\Lambda+m\right)^{2-n-|w|}\tau^{-\frac{1}{2}}_j\\\times
\mathcal{P}\!\left(\log \frac{\Lambda+m}{m}\right)
\mathcal{P}\!\left(\frac{\left\|\vec{p}_n\right\|}{\Lambda+m}\right)
\mathcal{Q}\!\left(\frac{\tau^{-{\frac{1}{2}}}}{\Lambda+m}\right)
\mathcal{F}^{\Lambda}_{s,l}(\tau)~.
\end{eqnarray*}
A possible choice for $\delta$ is $\sqrt{2}-1$. All the constants are absorbed in the polynomial $\mathcal{Q}\,$.\\
(D)  To prove (D) we use (\ref{R2}) and (\ref{R1*})
to show inductively that 
\begin{eqnarray}
\partial_{\Lambda}\partial^w 
\mathcal{L}_{l,n}^{\Lambda,\Lambda_0}
(z_1;p_1,\cdots,p_n;\left(z_1-z_i\right)^r\phi_{\tau,y_{2,s}})\ ,
\label{lrt}\\
~~~~~\partial_{\Lambda}F_{12} 
\mathcal{L}_{l,n}^{\Lambda,\Lambda_0}(z_1,z_2;p_{1,n};\phi_{\tau,y_{2,s}})\label{lrt2}
\end{eqnarray}
are $\mathcal{C}^{\infty}\left(\mathbb{R}^+\right)$ w.r.t. $z_1$. For (\ref{lrt2}), we will integrate from $\Lambda$ to $\Lambda_0$ and for (\ref{lrt}) we integrate from $0$ to $\Lambda$ for $n+|w|+r\leq4$, and from $\Lambda$ to $\Lambda_0$ for $n+|w|+r\geq 5$. The details of these integrations can be deduced from (II,a) and (II,b).\\
\underline{(II) Integration of the FE:}\\
 From the bounds on the derivatives $\partial_{\Lambda}\partial^{w}\mathcal{L}^{\Lambda,\Lambda_0}_{n+2,l-1}(z_1;p_1,\cdots,p_n;\phi_{\tau,y_{2,s}})$ we verify the induction hypothesis on integration over $\Lambda$. In all cases we need the bound
 $$\mathcal{F}_{s,l}^{\Lambda_2}(\tau)\leq \mathcal{F}_{s,l}^{\Lambda_1}(\tau) ~~~~~\mathrm{for}~~\Lambda_1\leq \Lambda_2$$ 
 which follows directly from the definition of
$\mathcal{F}_{s,l}^{\Lambda}(\tau)\,$.\\
\underline{a) Irrelevant terms:}
 \\
 Since 
 $$\partial^w\mathcal{L}^{\Lambda_0,\Lambda_0}_{l,n}((z_1,p_1),\cdots,(z_n,p_n))=0
~~~~~~\forall n+|w|+r\geq 5,$$
integration  from $\Lambda$ to $\Lambda_0$ for $n+|w|+r\geq 5\,$
gives 
 \begin{eqnarray*}
\fl |\partial^{w}\mathcal{L}^{\Lambda,\Lambda_0}_{l,n}(z_1;\vec{p}_n;(z_1-z_i)^r\phi_{\tau,y_{2,s}})|
\\\leq \left(\Lambda+m\right)^{4-n-|w|}\mathcal{P}\left(\log \frac{\Lambda+m}{m}\right)
   \mathcal{P}\left(\frac{\left\|\vec{p}_n\right\|}{\Lambda+m}\right) \mathcal{Q}\left(\frac{\tau^{-\frac{1}{2}}}{\Lambda+m}\right)\mathcal{F}^{\Lambda}_{s,l}\left(\tau\right)~,
\end{eqnarray*}
and
\begin{eqnarray*}
    \fl \left |\partial^w\mathcal{L}_{l,n}^{\Lambda,\Lambda_0}(z_1;\vec{p}_n;\phi_{\tau,y_{2,s}}^{(j)})\right|\nonumber\\\fl \indent ~~~~~~~\leq \left(\Lambda+m\right)^{3-n-|w|}\tau^{-\frac{1}{2}}_j\mathcal{P}_5\left(\log \frac{\Lambda+m}{m}\right)\mathcal{P}_6\left(\frac{\left\|\vec{p}_n\right\|}{\Lambda+m}\right)\mathcal{Q}_2\left(\frac{\tau^{-{\frac{1}{2}}}}{\Lambda+m}\right)\mathcal{F}^{\Lambda}_{s,l}(\tau)\ ,
\end{eqnarray*}
and 
\begin{eqnarray*}
\fl\left| F_{12} \mathcal{L}_{l,n}^{\Lambda,\Lambda_0}(z_1,z_2;\vec{p}_n;\phi_{\tau,y_{2,s}})\right|\leq \left(\Lambda+m\right)^{1-n}\mathcal{P}_3\left(\log \frac{\Lambda+m}{m}\right)\mathcal{P}_4\left(\frac{\left\|\vec{p}_n\right\|}{\Lambda+m}\right) \mathcal{F}^{12}_{s,l}(\Lambda,\tau)\ .
\end{eqnarray*}
 \underline{(b) Relevant terms:}\\
\underline{($b_1$) $n=4,w=0,r=0$:} We start from the decomposition (\ref{f3})
\begin{equation}\label{decoo}
\fl ~~~\mathcal{L}_{l,4}^{\Lambda,\Lambda_0}(z_1;0,\cdots,0;\Phi_{4,s})=c_{l}^{\Lambda,\Lambda_0}(z_1)\Phi_{4,s}(z_1,z_1,z_1)+l_{l,4}^{\Lambda,\Lambda_0}(z_1,\Phi_{4,s})~,~~s\leq 4~,
\end{equation}
 where 
 $$c_{l}^{\Lambda,\Lambda_0}(z_1):=\int_{0}^{\infty}dz_2dz_3dz_4\ 
\mathcal{L}^{\Lambda,\Lambda_0}_{l,4}\left((z_1,0),\cdots,(z_4,0)\right)$$
 and $$\Phi_{4,s}(z_2,z_3,z_4)=\prod_{i=2}^4 \phi_i(z_i),~~~\phi_i(z_i)=p_R\left(\tau_i;z_i,y_i\right)~~\mathrm{if~i\leq s,~~otherwise~~~}\phi_i\equiv1\ .$$
From the renormalization conditions we have for all $l\geq 1$
 $$c_{l}^{0,\Lambda_0}(z_1)\equiv0~.$$
Integrating (\ref{borne}) from $0$ to $\Lambda$
at zero momenta  then gives
 \begin{equation*}
|c_{l}^{\Lambda,\Lambda_0}(z_1)|
\leq \mathcal{P}\left(\log \frac{\Lambda+m}{m}\right)\   .
\end{equation*}
We decompose the test function 
\begin{equation}\label{etoile}
\fl \indent \Phi_{4,s}(z_2,z_3,z_4)
=\prod_{i=2}^s p_R(\tau_i;z_i,y_i)=\Phi_{4,s}(z_1,z_1,z_1)+\psi(z_2,z_3,z_4)\ ,
\end{equation}
where for $s=4$
\begin{equation*}
\fl \indent \psi(z_2,z_3,z_4):=\sum_{j=2}^4\prod^{j-1}_{f=2}p_R(\tau_f;z_1,y_f)p_R^{(1)}(\tau_j;z_j,z_1;y_j)\prod_{i=j+1}^4 p_R(\tau_i;z_i,y_i)=
\sum_{j=2}^4\phi^{(j)}_{\tau,y_{2,s}}(z_{2,4})
\end{equation*}
remembering definition (\ref{phij}).
Note that if $\phi_i\equiv1$ for some $i$ which corresponds to $s<4$, 
then the corresponding contribution to the sum vanishes.\\
Using (\ref{etoile}) and the linearity of 
$\mathcal{L}_{l,n}^{\Lambda,\Lambda_0}$ w.r.t. to the test functions, 
we deduce that
\begin{equation*}
\mathcal{L}_{l,4}^{\Lambda,\Lambda_0}(z_1;0,\cdots,0;\Phi_{4,s})
=c_{l}^{\Lambda,\Lambda_0}(z_1)\Phi_{4,s}(z_1,z_1,z_1)
+\mathcal{L}_{l,4}^{\Lambda,\Lambda_0}(z_1;0,\cdots,0;\psi)\ .
\end{equation*}
Therefore 
we have $l_{l,4}^{\Lambda,\Lambda_0}(z_1;\Phi_{4,s})
=\mathcal{L}_{l,4}^{\Lambda,\Lambda_0}(z_1;0,\cdots,0;\psi)\,$, 
and hence the FE (\ref{FEL}) provides 
\begin{eqnarray}\label{phiij}
   \fl \partial_{\Lambda}l_{l,4}^{\Lambda,\Lambda_0}\left(z_1;\Phi_{4,s}\right)=\frac{1}{2}\int_{z_{2,4},z,z'}  \psi(z_2,z_3,z_4)\\\left[\int_k\mathcal{L}_{l-1,6}^{\Lambda,\Lambda_0}\left((z,k),(z_1,0),\cdots,(z_{4},0),(z',-k)\right)
\dot{C}^{\Lambda}_R(k;z,z')\right.\nonumber\\
    -\frac{1}{2} \sum_{l_1+l_2=l}\sum_{n_1+n_2=4} \left[\mathcal{L}_{l_1,n_1+1}^{\Lambda,\Lambda_0}((z_1,0),\cdots,(z_{n_1},0),(z,0))
\dot{C}^{\Lambda}_R(0;z,z')\right.\nonumber\\\left.\left.\mathcal{L}_{l_2,n_2+1}^{\Lambda,\Lambda_0}((z',0),\cdots,(z_{4},0))\right]_{rsym}\right].
\end{eqnarray}
The r.h.s. is a sum over expressions of the same form as the one for $\partial_{\Lambda}\mathcal{L}_{l,4}^{\Lambda,\Lambda_0}(z_1,\phi_{\tau_{2,s},y_{2,s}}^{(j)})$
 in part (C). We obtain in the same way as there the bound 
\begin{equation*}
    \left|\partial_{\Lambda}l_{l,4}^{\Lambda,\Lambda_0}\left(z_1;\Phi_{4,s}\right)\right|\leq \left(\Lambda+m\right)^{-2}\tau^{-\frac{1}{2}}\mathcal{P}\left(\log \frac{\Lambda+m}{m}\right)\mathcal{Q}\left( \frac{\tau^{-\frac{1}{2}}}{\Lambda+m}\right)\mathcal{F}_{s,l}^{\Lambda}(\tau)~.
\end{equation*}
Integrating from $\Lambda$ to $\Lambda_0$ and majorizing $\left(\lambda+m\right)^{-1}$ by $\left(\Lambda+m\right)^{-1}$ we obtain 
\begin{equation*}
    \left|l_{l,4}^{\Lambda,\Lambda_0}\left(z_1;\Phi_{4,s}\right)\right|\leq \left(\frac{\tau^{-\frac{1}{2}}}{\Lambda+m}\right)\mathcal{P}\left(\log \frac{\Lambda+m}{m}\right)\mathcal{Q}\left( \frac{\tau^{-\frac{1}{2}}}{\Lambda+m}\right)\mathcal{F}_{s,l}^{\Lambda}(\tau)
\end{equation*}
which gives the bound for $l_{l,4}^{\Lambda,\Lambda_0}(z_1,\Phi_{4,s})\,$.
The bound  is extended to general momenta using 
the Taylor formula (\ref{taylorfor}).\\
\underline{($b_2$) $n=2,r=0,w=0$}: We start from the decomposition (\ref{f1})
\begin{eqnarray}\label{sum0}
   \fl \indent \mathcal{L}_{l,2}^{\Lambda,\Lambda_0}\left(z_1;0,0;\phi_2\right)=a_l^{\Lambda,\Lambda_0}(z_1)\phi_2(z_1)-s_l^{\Lambda,\Lambda_0}(z_1)(\partial_{z_1}\phi_2)(z_1)\nonumber\\
    +~d_l^{\Lambda,\Lambda_0}(z_1) (\partial_{z_1}^2\phi_2)(z_1)
    +l_{l,2}^{\Lambda,\Lambda_0}(z_1;\phi_2)~,
\end{eqnarray}
where $\phi_2(z_2):=\phi_{\tau,y_2}(z_2)=p_R(\tau;z_2,y_2)\,$. 
Using  $a_l^{0,\Lambda_0}(z_1),\ s_l^{0,\Lambda_0}(z_1),\ d_l^{0,\Lambda_0}(z_1)\equiv0\,$
we obtain on integration from $0$ to $\Lambda$
\begin{equation}\label{r1}
\fl ~~\left | a_l^{\Lambda,\Lambda_0}(z_1)\right|\leq (\Lambda+m)^2\mathcal{P}\left(\log \frac{\Lambda+m}{m}\right),~~~
    \left | s_l^{\Lambda,\Lambda_0}(z_1)\right|\leq (\Lambda+m)\mathcal{P}\left(\log \frac{\Lambda+m}{m}\right),
\end{equation}
\begin{equation}\label{r3}
    \left | d_l^{\Lambda,\Lambda_0}(z_1)\right|\leq \mathcal{P}\left(\log \frac{\Lambda+m}{m}\right).
\end{equation}
Since $(\partial_{z_1}\phi_2)(z_1)\leq \tau^{-\frac{1}{2}}\phi_2(z_1)$, $(\partial_{z_1}^2\phi_2)(z_1)\leq \tau^{-1}\phi_2(z_1)$ we obtain 
$$\left| s_l^{\Lambda,\Lambda_0}(z_1)(\partial_{z_1}\phi_2)(z_1)\right|\leq (\Lambda+m)^2\left(\frac{\tau^{-\frac{1}{2}}}{\Lambda+m}\right)\mathcal{P}\left(\log \frac{\Lambda+m}{m}\right)\phi_2(z_1)~,$$
$$\left| d_l^{\Lambda,\Lambda_0}(z_1)(\partial_{z_1}^2\phi_2)(z_1)\right|\leq (\Lambda+m)^2\left(\frac{\tau^{-\frac{1}{2}}}{\Lambda+m}\right)^2\mathcal{P}\left(\log \frac{\Lambda+m}{m}\right)\phi_2(z_1)~.$$
For the irrelevant part of the two-point function we have 
\begin{eqnarray}
\fl\partial_{\Lambda}l_{2,l}^{\Lambda,\Lambda_0}(z_1;\phi_2)\nonumber=\int_{0}^{\infty}dz_2\int_0^1 dt~ \frac{(1-t)^2}{2!}\partial_t^3{\phi}_{\tau,y_1}
\left(t z_2+(1-t)z_1\right)\partial_{\Lambda}
\mathcal{L}_{l,2}^{\Lambda,\Lambda_0}\left((z_1,0),(z_2,0)\right) \label{irr}
\\\fl~~~~=\int_{0}^{\infty}dz_2\int_0^1 dt ~\frac{(1-t)^2}{2!}
\frac{\partial_t^3{\phi}_2\left(tz_2+(1-t)z_1\right)}{(z_2-z_1)^3}
\partial_{\Lambda}F_{12}
\mathcal{L}_{l,2}^{\Lambda,\Lambda_0}\left((z_1,0),(z_2,0)\right).~\label{121}
\end{eqnarray}
The bound (\ref{12}) for $n=2,~r=0$ yields 
\begin{eqnarray*}
\fl \left|\partial_{\Lambda}l_{l,2}^{\Lambda,\Lambda_0}(z_1;\phi_2)\right|\leq  \left(\Lambda+m\right)^{-2}\mathcal{P}\left(\log \frac{\Lambda+m}{m}\right)\\ \times\int_{0}^{\infty}dz_2~\mathcal{F}^{12}_{s,l}(\Lambda,\tau)\int_0^1 dt \frac{(1-t)^2}{2!}\frac{\partial_t^3p_R\left(\tau;tz_2+(1-t)z_1,y_2\right)}{(z_2-z_1)^3}~.
\end{eqnarray*}
Using Lemma \ref{lemma2} we obtain 
\begin{eqnarray*}
\fl\left| \partial_{\Lambda}l_{2,l}^{\Lambda,\Lambda_0}(z_1;\phi_2)\right|\leq  ~O(1)~\left(\Lambda+m\right)^{-2}\tau^{-\frac{3}{2}}\mathcal{P}\left(\log \frac{\Lambda+m}{m}\right) \\
\times \int_{0}^{\infty}dz_2~\mathcal{F}^{12}_{s,l}(\Lambda,\tau)\int_0^1 dt ~p_R\left(\tau_{\delta'};tz_2+(1-t)z_1,y_2\right)~.
\end{eqnarray*}
Remembering (\ref{15}) and (\ref{1996}) we have 
\begin{eqnarray}
   \fl \mathcal{F}^{(12)}_{2,l}(\Lambda,\tau)
\noindent&=\mathcal{F}_{2,l}(\Lambda;z_1,z_2)
=\sum_{T^{2,(12)}}\mathcal{F}_{2,l}(\Lambda;T^{2,(12)}_l;z_1,z_2)\nonumber\\
     &=\sum_{n=1}^{3l-2}
\sup_{\left\{\Lambda_{I_{\nu}}|\Lambda\leq \Lambda_{I_{\nu}}\leq \Lambda_0\right\}}
\Biggl[\prod_{1\leq \nu \leq n}\int_{\tilde{z}_{\nu}}\Biggr]
p_B(\frac{1+\delta}{\Lambda_{I_1}^2};z_1,\tilde{z}_1)\cdots 
p_B(\frac{1+\delta}{\Lambda_{I_n}^2};\tilde{z}_{n},{z}_2)\label{ri}\\
    &\leq\sum_{n=1}^{3l-2}
\sup_{\left\{\Lambda_{I_{\nu}}|\Lambda\leq \Lambda_{I_{\nu}}\leq \Lambda_0\right\}} 
p_B(\frac{1+\delta}{\Lambda_{n}^2};z_1,z_2),\nonumber
\end{eqnarray}
where $\Lambda_n:=\left(\sum_{\nu=1}^n\Lambda_{I_{\nu}}^{-2}\right)^{-\frac{1}{2}}$. 
The sum (\ref{ri}) stems from the fact that the double rooted trees 
have all internal vertices with incidence number $2$ only. 
This  number $v_2$ is constrained by the relation 
$v_2+\delta_{c_1,1}\leq 3l-2+1/2$ from the definitions (iii)-(iv) 
of subsection \ref{ts}.\\
 Using Lemma \ref{lemme1} together with
 \begin{equation*}
     p_R(\tau_{\delta'};t z_2+(1-t)z_1,y_2)\leq 2~p_B(\tau_{\delta'};t z_2+(1-t)z_1,y_2)~.
 \end{equation*}
For $\delta'\Lambda^2
\geq  b\tau^{-1}$, we obtain
 \begin{equation}\label{120}
\fl~~~~\left|\partial_{\Lambda}l_{2,l}^{\Lambda,\Lambda_0}(z_1;\phi_2)\right|
\leq 
\tau^{-\frac{3}{2}}\left(\Lambda+m\right)^{-2}
\mathcal{P}\left(\log \frac{\Lambda+m}{m}\right)
   p_B((1+\delta')^3\tau;z_1,y_2)\ .
\end{equation}
As before, we choose $\delta'=\frac{\delta}{3}+O(\delta^2)$ such that $(1+\delta')^3=1+\delta$, which implies that
 \begin{equation}\label{122}
\fl~~~~\left|\partial_{\Lambda}l_{2,l}^{\Lambda,\Lambda_0}(z_1;\phi_2)\right|
\leq 
\tau^{-\frac{3}{2}}\left(\Lambda+m\right)^{-2}
\mathcal{P}\left(\log \frac{\Lambda+m}{m}\right)
   p_B(\tau_{\delta};z_1,y_2)\ .
\end{equation}
For $\delta'\Lambda^2\leq  b\tau^{-1}$ we
use 
\begin{eqnarray*}
    \fl \partial_{\Lambda}l_{l,2}^{\Lambda,\Lambda_0}(z_1;\phi_2)
=\partial_{\Lambda}\mathcal{L}_{l,2}^{\Lambda,\Lambda_0}\left(z_1;0,0;\phi_2\right)
-\partial_{\Lambda}a_l^{\Lambda,\Lambda_0}(z_1)\phi_2(z_1) 
- \partial_{\Lambda}s_l^{\Lambda,\Lambda_0}(z_1)(\partial_{z_1}\phi_2)(z_1)\\
\indent-\partial_{\Lambda}d_l^{\Lambda,\Lambda_0}(z_1) (\partial_{z_1}^2\phi_2)(z_1)
\end{eqnarray*}
and the bounds (\ref{lam0}) and  (\ref{brel}) to (\ref{bre}) to obtain
\begin{equation}\label{123}
\fl~~~~\left|\partial_{\Lambda}l_{2,l}^{\Lambda,\Lambda_0}(z_1;\phi_2)\right|
\leq \left(\Lambda+m\right)\mathcal{P}\left(\log \frac{\Lambda+m}{m}\right)\mathcal{Q}\left(\frac{\tau^{-\frac{1}{2}}}{\Lambda+m}\right)
   p_B(\tau;z_1,y_2)\ .
\end{equation}
Since $\delta'\Lambda^2\leq b\tau^{-1}\,$ we have
\begin{eqnarray}\label{1222}
    \fl\left|\partial_{\Lambda}l_{l,2}^{\Lambda,\Lambda_0}(z_1;\phi_2)\right|
\leq\max\Bigl(m^3,\tau^{-\frac{3}{2}}\left(\frac{b}{\delta'}\right)^{\frac{3}{2}}\Bigr)
\left(\Lambda+m\right)^{-2}\nonumber\\\mathcal{P}\left(\log \frac{\Lambda+m}{m}\right)\mathcal{Q}\left(\frac{\tau^{-\frac{1}{2}}}{\Lambda+m}\right)
   p_B({\tau};z_1,y_2)\ .
\end{eqnarray}
Combining (\ref{122}) and (\ref{1222}) and using the bound (\ref{pbdelta}), we obtain for all $\Lambda\geq 0$,
\begin{eqnarray}
    \fl\left|\partial_{\Lambda}l_{l,2}^{\Lambda,\Lambda_0}(z_1;\phi_2)\right|
\leq\max\Bigl(m^3,\tau^{-\frac{3}{2}}\left(\frac{b}{\delta'}\right)^{\frac{3}{2}}\Bigr)
\left(\Lambda+m\right)^{-2}\nonumber\\\mathcal{P}\left(\log \frac{\Lambda+m}{m}\right)\mathcal{Q}\left(\frac{\tau^{-\frac{1}{2}}}{\Lambda+m}\right)
   p_B({\tau}_{\delta};z_1,y_2)\ .
\end{eqnarray}
Integrating from $\Lambda$ to $\Lambda_0$ gives 
\begin{eqnarray*}
\fl\left|l_{2,l}^{\Lambda,\Lambda_0}(z_1,\phi_2)\right|&\leq  
\max \left(m^3,\tau^{-\frac{3}{2}}\left(\frac{b}{\delta'}\right)^{\frac{3}{2}}\right)(\Lambda+m)^{-1}
\mathcal{P}\left(\log \frac{\Lambda+m}{m}\right)\mathcal{Q}\left(\frac{\tau^{-\frac{1}{2}}}{\Lambda+m}\right)
   p_B({\tau}_{\delta};z_1,y_2)\\
   &\leq \frac{\max \left(m^3,\tau^{-\frac{3}{2}}\left(\frac{b}{\delta'}\right)^{\frac{3}{2}}\right)}
{(\Lambda+m)^3}(\Lambda+m)^2\mathcal{P}\left(\log \frac{\Lambda+m}{m}\right)\mathcal{Q}\left(\frac{\tau^{-\frac{1}{2}}}{\Lambda+m}\right)
  p_B({\tau}_{\delta};z_1,y_2)\\
   &\leq (\Lambda+m)^2\mathcal{P}\left(\log \frac{\Lambda+m}{m}\right)
   \tilde{\mathcal{Q}}\left(\frac{\tau^{-\frac{1}{2}}}{\Lambda+m}\right)\ 
p_B({\tau}_{\delta};z_1,y_2)~,
\end{eqnarray*}
where $\tilde{\mathcal{Q}}(x)=(1+x^3)\mathcal{Q}(x)$. Again, all the constants were absorbed in the coefficients of $\mathcal{P}$. This 
concludes the proof for $n=2, r=0$ and $w=0$.\\
To establish the bounds on $\partial^w 
\mathcal{L}_{l,2}^{\Lambda,\Lambda_0}(z_1;p_{1,n};(z_1-z_2)^r\phi_2)$ 
for $r=1, 2;~w=0$ and $r=0;~w=2$,  we expand the respective test 
functions as follows,
\begin{eqnarray*}
    \mathcal{L}_{l,2}^{\Lambda,\Lambda_0}(z_1;0,0;(z_1-z_2)\phi_2)
=-s_l^{\Lambda,\Lambda_0}(z_1)\phi_2(z_1)+
d_l^{\Lambda,\Lambda_0}(z_1) (\partial_{z_1}\phi_2)(z_1)\\
+\int_0^{\infty}dz_2\int_0^1 dt (1-t)\partial_t^2{\phi_2}
\left(t z_2+(1-t)z_1\right)(z_1-z_2)
\mathcal{L}_{l,2}^{\Lambda,\Lambda_0}\left((z_1,0),(z_2,0)\right),
\end{eqnarray*}
\begin{eqnarray*}
    \mathcal{L}_{l,2}^{\Lambda,\Lambda_0}(z_1;0,0;(z_1-z_2)^2\phi_2)
=d_l^{\Lambda,\Lambda_0}(z_1)\phi_2(z_1)
    \\+\int_0^{\infty}dz_2\int_0^1 dt~ \partial_t{\phi_2}\left(t z_2+(1-t)z_1
\right)(z_1-z_2)^2\mathcal{L}_{l,2}^{\Lambda,\Lambda_0}\left((z_1,0),(z_2,0)\right),
\end{eqnarray*}
\begin{eqnarray*}\label{f2}
    \left(\partial_{p^2}\mathcal{L}_{l,2}^{\Lambda,\Lambda_0}\right)
\left(z_1;0,0;\phi_2\right)=b_l^{\Lambda,\Lambda_0}(z_1)\phi_2(z_1)
    \\+\int_{0}^{\infty}dz_2\int_0^1 dt \frac{(1-t)^2}{2!}
\partial_t^3{\phi_2}\left(t z_2+(1-t)z_1\right)
\partial_{p^2}\left(\mathcal{L}_{l,2}^{\Lambda,\Lambda_0}
\left((z_1,p),(z_2,-p)\right)\right)_{|_{p=0}}.
\end{eqnarray*}
The relevant terms have been bounded in (\ref{r1})-(\ref{r3}), 
and the reminders are treated as 
$l_{2,l}^{\Lambda,\Lambda_0}(z_1;\phi_2)\,$. We  obtain
\begin{eqnarray*}
    \left| \mathcal{L}_{l,2}^{\Lambda,\Lambda_0}(z_1;0;(z_1-z_2)\phi_{2})\right|
\leq \left(\Lambda+m\right)\mathcal{P}\left(\log \frac{\Lambda+m}{m}\right)
\mathcal{Q}\left(\frac{\tau^{-\frac{1}{2}}}{\Lambda+m}\right)
\mathcal{F}^{\Lambda}_{s,l}(\tau)~,\\
\left| \mathcal{L}_{l,2}^{\Lambda,\Lambda_0}(z_1;0;(z_1-z_2)^2\phi_{2})\right|
\leq \mathcal{P}\left(\log \frac{\Lambda+m}{m}\right)
\mathcal{Q}\left(\frac{\tau^{-\frac{1}{2}}}{\Lambda+m}\right)
\mathcal{F}^{\Lambda}_{s,l}(\tau)~,\\
\left| \partial^2_{p}\mathcal{L}_{l,2}^{\Lambda,\Lambda_0}(z_1;0;\phi_{2})\right|
\leq \mathcal{P}\left(\log \frac{\Lambda+m}{m}\right)
\mathcal{Q}\left(\frac{\tau^{-\frac{1}{2}}}{\Lambda+m}\right)
\mathcal{F}^{\Lambda}_{s,l}(\tau)~.
\end{eqnarray*}
The extension to general momenta is performed using the Taylor 
expansion of $\mathcal{L}_{l,2}^{\Lambda,\Lambda_0}(z_1;0,0;(z_1-z_2)^r\phi_{2})$ 
for $r=0,1,2$ w.r.t. the variable $p\in \mathbb{R}^3$. \\
Finally, note that for
\begin{equation*}
    \phi_{2}^{(2)}(z_2)=p_R(\tau;z_2,y_2)-p_R(\tau;z_1,y_2)=\phi_2(z_2)-\phi_2(z_1)\ ,
\end{equation*}
we have 
\begin{eqnarray}
   \fl ~~ \mathcal{L}_{l,2}^{\Lambda,\Lambda_0}
\left(z_1;0,0;\phi_{2}^{(2)}\right)
=s_l^{\Lambda,\Lambda_0}(z_1)(\partial_{z_1}\phi_2)(z_1)
    +d_l^{\Lambda,\Lambda_0}(z_1) (\partial_{z_1}^2\phi_2)(z_1)
    +l_{l,2}^{\Lambda,\Lambda_0}(z_1;0,0;\phi_2)~.
\end{eqnarray}
Proceeding again similarly as before - 
see (\ref{121}), (\ref{122}) and (\ref{123}) - provides
\begin{equation*}
    \left |\mathcal{L}_{l,2}^{\Lambda,\Lambda_0}(z_1;0,0;\phi_{2}^{(2)})\right|
\leq \left(\Lambda+m\right)\tau^{-\frac{1}{2}}
\mathcal{P}\left(\log \frac{\Lambda+m}{m}\right)
\mathcal{Q}\left(\frac{\tau^{-{\frac{1}{2}}}}{\Lambda+m}\right)
\mathcal{F}^{\Lambda}_{s,l}(\tau)\ .
\end{equation*}
The extension to general momenta is done by Taylor expansion. 
This ends the proof of Theorem 1.
\end{proof}
\begin{theorem}(Convergence)
Let $0\leq\Lambda\leq \Lambda_0<\infty$. Using the same notations, conventions and adopting the same renormalization conditions (\ref{renoc}) as in Theorem 1,
 we have the following bounds
\begin{eqnarray}\label{cc1}
    \fl \left| \partial_{\Lambda_0}\partial^w \mathcal{L}_{l,n}^{\Lambda,\Lambda_0}(z_1;\vec{p}_n;\left(z_1-z_i\right)^r\phi_{\tau,y_{2,s}})\right|\nonumber\leq \frac{\left(\Lambda+m\right)^{5-n-|w|-r}}{\left(\Lambda_0+m\right)^2}{\tilde\mathcal{P}}_1\left(\log \frac{\Lambda_0+m}{m}\right)\\\times\tilde{\mathcal{P}}_2\left(\frac{\left\|\vec{p}_n\right\|}{\Lambda+m}\right) \tilde{\mathcal{Q}}_1\left(\frac{\tau^{-\frac{1}{2}}}{\Lambda+m}\right)\mathcal{F}^{\Lambda}_{s,l}(\tau)~~~~\forall 
n+|w|+r\geq 4\ ,
\end{eqnarray}
\begin{eqnarray}\label{ccc1}
    \fl \left| \partial_{\Lambda_0}\partial^w \mathcal{L}_{l,2}^{\Lambda,\Lambda_0}(z_1;p,-p;\phi_{\tau,y_{2}})\right|\nonumber\\\leq \frac{\left(\Lambda+m\right)^{3-|w|}}{\left(\Lambda_0+m\right)^2}\tilde{\mathcal{P}}_3\left(\log \frac{\Lambda_0+m}{m}\right)\tilde{\mathcal{P}}_4\left(\frac{\left\|\vec{p}_n\right\|}{\Lambda+m}\right) \tilde{\mathcal{Q}}_2\left(\frac{\tau^{-\frac{1}{2}}}{\Lambda+m}\right)\mathcal{F}^{\Lambda}_{s,l}(\tau)\ ,
\end{eqnarray}
\begin{eqnarray}\label{cc2}
    \fl \left| \partial_{\Lambda_0} F_{12} \mathcal{L}_{l,n}^{\Lambda,\Lambda_0}(z_1,z_2;\vec{p}_n;\phi_{\tau,y_{2,s}})\right|\nonumber\\\leq \frac{\left(\Lambda+m\right)^{2-n}}{\left(\Lambda_0+m\right)^2}
\tilde{\mathcal{P}}_5\!\left(\log \frac{\Lambda_0+m}{m}\right)
\tilde{\mathcal{P}}_6\!\left(\frac{\left\|\vec{p}_n\right\|}{\Lambda+m}\right)
 \mathcal{F}^{12}_{s,l}(\Lambda,\tau)~~~~\forall n\geq 2\ ,\quad
\end{eqnarray}
\begin{eqnarray}\label{cc3}
    \fl \left |\partial_{\Lambda_0}\mathcal{L}_{l,n}^{\Lambda,\Lambda_0}(z_1;\vec{p}_n;\phi_{\tau,y_{2,s}}^{(j)})\right|\leq \frac{\left(\Lambda+m\right)^{4-n}}{\left(\Lambda_0+m\right)^2}\tau^{-\frac{1}{2}}_j\tilde{\mathcal{P}}_7\left(\log \frac{\Lambda_0+m}{m}\right)\nonumber\\\times\tilde{\mathcal{P}}_8\left( \frac{\left\|\vec{p}_n\right\|}{\Lambda+m}\right)\tilde{\mathcal{Q}}_3\left(\frac{\tau^{-{\frac{1}{2}}}}{\Lambda+m}\right)\mathcal{F}^{\Lambda}_{s,l}(\tau)~~~~\forall 
n\geq 4\ .
\end{eqnarray}
\end{theorem}
\begin{proof}
We apply the method of the previous proof. The case $n=4$, $l=0$ evidently satisfies the claim (\ref{cc1}). We integrate the system of flow equations (\ref{FEL}) and derive the individual $n$-point folded distributions (\ref{sh1}) w.r.t. $\Lambda_0$. We denote the r.h.s. of (\ref{FEL}) by $\partial^w \mathcal{R}^{\Lambda,\Lambda_0}_{l,n}\left((z_1,p_1),\cdots,(z_n,p_n)\right)$.
 We bound separately the relevant and the irrelevant terms.
\begin{itemize}
    \item \underline{I) $n+r+|w|>4$:}\\
 In these cases the boundary condition (\ref{loup}) imply
    \begin{equation*}
        \fl -\partial^w 
\mathcal{L}_{l,n}^{\Lambda,\Lambda_0}
\left(z_1;\vec{p}_n;(z_1-z_i)^r\phi_{\tau,y_{2,s}}\right)
=\int_{\Lambda}^{\Lambda_0}d\lambda 
\ \partial^w \mathcal{R}^{\lambda,\Lambda_0}_{l,n}
\left(z_1;\vec{p}_n;(z_1-z_i)^r\phi_{\tau,y_{2,s}}\right)\ .
    \end{equation*}
Therefore
\begin{eqnarray}\label{loupe}
\fl -\partial_{\Lambda_0}\partial^w \mathcal{L}_{l,n}^{\Lambda,\Lambda_0}\left(z_1;\vec{p}_n;(z_1-z_i)^r\phi_{\tau,y_{2,s}}\right)=\partial^w \mathcal{R}^{\Lambda_0,\Lambda_0}_{l,n}\left(z_1;\vec{p}_n;(z_1-z_i)^r\phi_{\tau,y_{2,s}}\right)\nonumber\\
+\int_{\Lambda}^{\Lambda_0}d\lambda \ \partial_{\Lambda_0}\partial^w \mathcal{R}^{\lambda,\Lambda_0}_{l,n}\left(z_1;\vec{p}_n;(z_1-z_i)^r\phi_{\tau,y_{2,s}}\right)\ .
\end{eqnarray}
To the first term on the r.h.s. only the non-linear part 
on the r.h.s. of (\ref{FEL}) contributes 
because of the boundary condition (\ref{loup}). 
Using Theorem $1$ we obtain as before the bound
\begin{eqnarray*}
\fl \left|\partial^w \mathcal{R}^{\Lambda_0,\Lambda_0}_{l,n}
\left(z_1;\vec{p}_n;(z_1-z_i)^r\phi_{\tau,y_{2,s}}\right)\right|\\
\leq \left(\Lambda_0+m\right)^{3-n-|w|-r}
\mathcal{P}\left(\log\frac{\Lambda_0+m}{m}\right)
\mathcal{P}\left(\frac{\left\|\vec{p}_n\right\|}{\Lambda+m}\right)
\mathcal{Q}\left(\frac{\tau^{-\frac{1}{2}}}{\Lambda+m}\right)
\mathcal{F}_{s,l}^{\Lambda}\left(\tau\right)\ .\nonumber
\end{eqnarray*}
Since $n+|w|+r>4$, we obtain for all $0\leq \Lambda\leq \Lambda_0$ 
\begin{eqnarray*}
\fl \left|\partial^w \mathcal{R}^{\Lambda_0,\Lambda_0}_{l,n}\left(z_1;\vec{p}_n;(z_1-z_i)^r\phi_{\tau,y_{2,s}}\right)\right|\\
\leq \frac{\left(\Lambda+m\right)^{5-n-|w|-r}}{\left(\Lambda_0+m\right)^2}\mathcal{P}\left(\log\frac{\Lambda_0+m}{m}\right)\mathcal{P}\left(\frac{\left\|\vec{p}_n\right\|}{\Lambda+m}\right)\mathcal{Q}\left(\frac{\tau^{-\frac{1}{2}}}{\Lambda+m}\right)\mathcal{F}_{s,l}^{\Lambda}\left(\tau\right)\ .\nonumber
\end{eqnarray*}
For the second term on the r.h.s. of (\ref{loupe}) we obtain  
\begin{eqnarray}
\fl \partial_{\Lambda_0}\partial^w 
\mathcal{R}^{\Lambda,\Lambda_0}_{l,n}\left((z_1,p_1),\cdots,(z_n,p_n)\right)\\
\fl=\frac{1}{2}
\int_z \int_{z'}
\int_k \partial_{\Lambda_0}
\partial^w\mathcal{L}_{l-1,n+2}^{\Lambda,\Lambda_0}
\left((z_1,p_1),\cdots,(z_n,p_n),(z,k),(z',-k)\right)
\dot{C}^{\Lambda}_R(p;z,z')\\
   \fl \indent -\frac{1}{2}\int_z \int_{z'} \sum_{l_1,l_2}'
\sum_{n_1,n_2}'\sum_{w_i}c_{w_i} \left[\partial_{\Lambda_0}\partial^{w_1}
\mathcal{L}_{l_1,n_1+1}^{\Lambda,\Lambda_0}((z_1,p_1),\cdots,(z_{n_1},p_{n_1}),(z,p))\ 
\partial^{w_3}\dot{C}^{\Lambda}_R(p;z,z')\right.\nonumber\\
\left.
\times\ 
\partial^{w_2}
\mathcal{L}_{l_2,n_2+1}^{\Lambda,\Lambda_0}((z',-p),\cdots,(z_{n},p_{n}))
\right]_{rsym}\nonumber\\
\fl \indent -\frac{1}{2}\int_z \int_{z'} \sum_{l_1,l_2}'
\sum_{n_1,n_2}'\sum_{w_i}c_{w_i} \left[\partial^{w_1}
\mathcal{L}_{l_1,n_1+1}^{\Lambda,\Lambda_0}((z_1,p_1),\cdots,(z_{n_1},p_{n_1}),(z,p))
\partial^{w_3}\dot{C}^{\Lambda}_R(p;z,z')
\right.\nonumber\\\left.
\times\
\partial_{\Lambda_0}\partial^{w_2}
\mathcal{L}_{l_2,n_2+1}^{\Lambda,\Lambda_0}((z',-p),\cdots,(z_{n},p_{n}))
\right]_{rsym},\nonumber\\
    p=-p_1-\cdots-p_{n_1}=p_{n_1+1}+\cdots+p_n\nonumber,
\end{eqnarray}
where we used that 
$\,\partial_{\Lambda_0}\dot{C}_R^{\Lambda}\left(k;z,z'\right)=0\,$.
Using Theorem 1 and the induction hypothesis (\ref{cc1}), 
and following the same steps as in the proof of Theorem 1 we get
\begin{eqnarray*}
\fl \left|\partial_{\Lambda_0}\partial^w 
\mathcal{R}^{\Lambda,\Lambda_0}_{l,n}\left(z_1;\vec{p}_n;(z_1-z_i)^r
\phi_{\tau,y_{2,s}}\right)\right|\\
\leq \frac{\left(\Lambda+m\right)^{4-n-|w|-r}}{\left(\Lambda_0+m\right)^2}\mathcal{P}\left(\log \frac{\Lambda_0+m}{m}\right)\mathcal{P}\left(\frac{\left\|\vec{p}_n\right\|}{\Lambda+m}\right)\mathcal{Q}\left(\frac{\tau^{-\frac{1}{2}}}{\Lambda+m}\right)\mathcal{F}_{s,l}^{\Lambda}\left(\tau\right)\ .\nonumber
\end{eqnarray*}
Integrating from $\Lambda$ to $\Lambda_0$ and using that
\begin{eqnarray}
 \frac{\left\|\vec{p}_n\right\|}{\lambda+m}\leq \frac{\left\|\vec{p}_n\right\|}{\Lambda+m}\ ,
~~~~\frac{\tau^{-\frac{1}{2}}}{\lambda+m}\leq 
\frac{\tau^{-\frac{1}{2}}}{\Lambda+m},~~~~\mathcal{F}_{s,l}^{\lambda}\left(\tau\right)
\leq \mathcal{F}_{s,l}^{\Lambda}\left(\tau\right)~~~\forall \Lambda\leq \lambda
\end{eqnarray}
gives a bound on the second term on the r.h.s. of (\ref{loupe}), which 
is of the type (\ref{cc1}).\\
(\ref{cc2}) and (\ref{cc3}) are proved following the same steps.
\item \underline{II) $(n,r,|w|)=(4,0,0)$, $(n,r,|w|)=(2,0,0)$ and $(n,r,|w|)=(2,0,2)$:}\\
 The FE equation provides inductive bounds on the relevant parts in these cases, and since the renormalization conditions do not depend on $\Lambda_0$, we obtain
\begin{eqnarray}
    |\partial_{\Lambda_0}c_l^{\Lambda,\Lambda_0}(z_1)|\leq \frac{\left(\Lambda+m\right)}{\left(\Lambda_0+m\right)^2}\mathcal{P}\left(\log \frac{\Lambda_0+m}{m}\right),\label{clam}\\|\partial_{\Lambda_0}a_l^{\Lambda,\Lambda_0}(z_1)|\leq \frac{\left(\Lambda+m\right)^3}{\left(\Lambda_0+m\right)^2}\mathcal{P}\left(\log \frac{\Lambda_0+m}{m}\right),~~~~~~\label{5}\\
     |\partial_{\Lambda_0}b_l^{\Lambda,\Lambda_0}(z_1)|\leq \frac{\left(\Lambda+m\right)}{\left(\Lambda_0+m\right)^2}\mathcal{P}\left(\log \frac{\Lambda_0+m}{m}\right),\\|\partial_{\Lambda_0}d_l^{\Lambda,\Lambda_0}(z_1)|\leq \frac{\left(\Lambda+m\right)}{\left(\Lambda_0+m\right)^2}\mathcal{P}\left(\log \frac{\Lambda_0+m}{m}\right),~~~~~~\\
    |\partial_{\Lambda_0}s_l^{\Lambda,\Lambda_0}(z_1)|\leq \frac{\left(\Lambda+m\right)^2}{\left(\Lambda_0+m\right)^2} \mathcal{P}\left(\log \frac{\Lambda_0+m}{m}\right)~. \label{55}
\end{eqnarray}
In the case $n=4$, we use the decomposition (\ref{decoo}) together with (\ref{phij}), (\ref{cc3}) and (\ref{clam}) to obtain the bound 
\begin{eqnarray}\label{cc1'}
    \fl \left| \partial_{\Lambda_0}\partial^w \mathcal{L}_{l,4}^{\Lambda,\Lambda_0}(z_1;\vec{0};\phi_{\tau,y_{2,s}})\right|\nonumber\leq \frac{\left(\Lambda+m\right)}{\left(\Lambda_0+m\right)^2}\tilde{\mathcal{P}}_1\left(\log \frac{\Lambda_0+m}{m}\right)\\\times\tilde{\mathcal{P}}_2\left(\frac{\left\|\vec{p}_n\right\|}{\Lambda+m}\right) \tilde{\mathcal{Q}}_1\left(\frac{\tau^{-\frac{1}{2}}}{\Lambda+m}\right)\mathcal{F}^{\Lambda}_{s,l}(\tau)\ .
\end{eqnarray}
For $n=2$, we use the decomposition (\ref{sum0}) and follow the same steps as in part b2) of the proof of Theorem 1. Using (\ref{irr}) and the bound (\ref{cc2}), we obtain for all $0<\delta'<1$,
\begin{eqnarray*}
\fl  \partial_{\Lambda_0}
l_{2,l}^{\Lambda,\Lambda_0}(z_1;\phi_2)\nonumber=
\int_{0}^{\infty}dz_2\int_0^1 dt \frac{(1-t)^2}{2!}
\frac{\partial_t^3{\phi}_2\left(tz_2+(1-t)z_1\right)}{(z_2-z_1)^3}
\partial_{\Lambda_0}F_{12}
\mathcal{L}_{l,2}^{\Lambda,\Lambda_0}\left((z_1,0),(z_2,0)\right)\\
\fl \indent \leq  {\left(\Lambda_0+m\right)^{-2}}
\mathcal{P}
\left(\log \frac{\Lambda_0+m}{m}\right) 
\int_{0}^{\infty}dz_2
\ \mathcal{F}^{12}_{s,l}(\Lambda,\tau)
\int_0^1 dt~ p_R\left(\tau_{\delta'};tz_2+(1-t)z_1,y_2\right).
\end{eqnarray*}
Following the same steps as used before we obtain 
\begin{equation}\label{fia}
\fl~~~~~~\left|\partial_{\Lambda_0}l_{2,l}^{\Lambda,\Lambda_0}(z_1;\phi)\right|\leq \tau^{-\frac{3}{2}}
{(\Lambda_0+m)^{-2}}
\mathcal{P}\left(\log \frac{\Lambda_0+m}{m}\right)
   p_B(\tau_{\delta};z_1,y_2)~~~ \forall \delta'\Lambda^2\geq  b\tau^{-1}\ ,
\end{equation}
where we choose as before $\delta'=\delta/3+O(\delta^2)$.\\
For the case $\delta'\Lambda^2\leq b\tau^{-1}$, the bound is obtained as in the proof of Theorem 1 by using the decomposition (\ref{sum0}) which yields
\begin{eqnarray*}
    \fl \partial_{\Lambda}\partial_{\Lambda_0}l_{l,2}^{\Lambda,\Lambda_0}(z_1;\phi_2)=\partial_{\Lambda}\partial_{\Lambda_0}\mathcal{L}_{l,2}^{\Lambda,\Lambda_0}\left(z_1;0,0;\phi_2\right)-\partial_{\Lambda}\partial_{\Lambda_0}a_l^{\Lambda,\Lambda_0}(z_1)\phi_2(z_1) \\+\partial_{\Lambda}\partial_{\Lambda_0}s_l^{\Lambda,\Lambda_0}(z_1)(\partial_{z_1}\phi_2)(z_1)-\partial_{\Lambda}\partial_{\Lambda_0}d_l^{\Lambda,\Lambda_0}(z_1) (\partial_{z_1}^2\phi_2)(z_1)\ .
\end{eqnarray*}
Using the induction hypothesis (\ref{cc1}), we obtain as in (\ref{123})
\begin{equation*}
    \fl~~~~\left|\partial_{\Lambda}
\partial_{\Lambda_0}l_{l,2}^{\Lambda,\Lambda_0}(z_1;\phi_2)\right|
\leq\frac{\left(\Lambda+m\right)^2}
{\left(\Lambda_0+m\right)^2}\mathcal{P}\left(\log \frac{\Lambda_0+m}{m}\right)\mathcal{Q}\left(\frac{\tau^{-\frac{1}{2}}}{\Lambda+m}\right)
   p_B(\tau_{\delta};z_1,y_2)\ .
\end{equation*}
Since $\delta'\Lambda^2\leq b\tau^{-1}$ we find 
\begin{eqnarray}\label{cuicui}
\!
\fl \left|\partial_{\Lambda}
\partial_{\Lambda_0}l_{2,l}^{\Lambda,\Lambda_0}(z_1,\phi_2)\right|\nonumber\\
\fl~~~~~\leq \max\Bigl(m^4,{\tau^{-2}}\left(\frac{b}{\delta'}\right)^{2}\Bigr)
\frac{\left(\Lambda+m\right)^{-2}}
{\left(\Lambda_0+m\right)^{2}}
\mathcal{P}\Bigl(\log \frac{\Lambda_0+m}{m}\Bigr)
  \! \mathcal{Q}\Bigl(\frac{\tau^{-\frac{1}{2}}}{\Lambda+m}\Bigr)
p_B(\tau_{\delta};z_1,y_2)~.
\end{eqnarray}
Integrating from $\Lambda$ to $\Lambda_0$ (with $\Lambda_0$ large enough), we obtain
\begin{eqnarray*}
\fl~~~~~\partial_{\Lambda_0}l_{2,l}^{\Lambda,\Lambda_0}(z_1,\phi_2)=\int_{\Lambda}^{\frac{b}{\delta'}{\tau}^{-\frac{1}{2}}}d\lambda~\partial_{\lambda}\partial_{\Lambda_0}l_{2,l}^{\Lambda,\Lambda_0}(z_1,\phi_2)+\int_{\frac{b}{\delta'}{\tau}^{-\frac{1}{2}}}^{\Lambda_0}d\lambda~\partial_{\lambda}\partial_{\Lambda_0}l_{2,l}^{\Lambda,\Lambda_0}(z_1,\phi_2)~.
\end{eqnarray*}
Using (\ref{cuicui}) we obtain 
\begin{eqnarray}
\fl \left|\int_{\Lambda}^{\frac{b}{\delta'}{\tau}^{-\frac{1}{2}}}d\lambda~\partial_{\lambda}\partial_{\Lambda_0}l_{2,l}^{\Lambda,\Lambda_0}(z_1,\phi_2)\right|\nonumber\\\fl~~~~~~~\leq  \max \left(m^4,{\tau^{-2}}\left(\frac{b}{\delta'}\right)^{2}\right)\frac{(\Lambda+m)^{-1}}{\left(\Lambda_0+m\right)^{2}}\mathcal{P}\left(\log \frac{\Lambda_0+m}{m}\right)\mathcal{Q}\Bigl(\frac{\tau^{-\frac{1}{2}}}{\Lambda+m}\Bigr)
   p_B(\tau_{\delta};z_1,y_2)~,\nonumber\\
   \fl~~~~~~~\leq \frac{(\Lambda+m)^3}{\left(\Lambda_0+m\right)^2}\mathcal{P}\left(\log \frac{\Lambda_0+m}{m}\right)
  {\mathcal{Q}} \left(\frac{\tau^{-\frac{1}{2}}}{\Lambda+m}\right)p_B(\tau_{\delta};z_1,y_2)~.\label{sir}
\end{eqnarray}
We have 
\begin{eqnarray}
     \int_{\frac{b}{\delta'}{\tau}^{-\frac{1}{2}}}^{\Lambda_0}d\lambda~\partial_{\lambda}\partial_{\Lambda_0}l_{2,l}^{\Lambda,\Lambda_0}(z_1,\phi_2)=\partial_{\Lambda_0}l_{2,l}^{\frac{b}{\delta'}\tau^{-\frac{1}{2}},\Lambda_0}(z_1,\phi_2)~.
\end{eqnarray}
Using (\ref{fia}) for $\delta'\Lambda^2\geq b\tau^{-1}$, we deduce that 
\begin{eqnarray}\label{sir1}
    \fl ~~~ \left|\int_{{\tau}^{-\frac{1}{2}}}^{\Lambda_0}d\lambda~\partial_{\lambda}\partial_{\Lambda_0}l_{2,l}^{\Lambda,\Lambda_0}(z_1,\phi_2)\right|\leq \frac{(\Lambda+m)^3}{\left(\Lambda_0+m\right)^2}\mathcal{P}\left(\log \frac{\Lambda_0+m}{m}\right)
   \left(\frac{\tau^{-\frac{1}{2}}}{\Lambda+m}\right)^3p_B(\tau_{\delta};z_1,y_2)~,
\end{eqnarray}
(\ref{sir}) together with (\ref{sir1}) imply that for all $0\leq \Lambda\leq \Lambda_0$, we have
\begin{eqnarray*}
    \fl ~~~ \left|\partial_{\Lambda_0}l_{2,l}^{\Lambda,\Lambda_0}(z_1,\phi_2)\right|\leq \frac{(\Lambda+m)^3}{\left(\Lambda_0+m\right)^2}\mathcal{P}\left(\log \frac{\Lambda_0+m}{m}\right)
   \mathcal{Q}\left(\frac{\tau^{-\frac{1}{2}}}{\Lambda+m}\right)~p_B(\tau_{\delta};z_1,y_2)~.
\end{eqnarray*}
\end{itemize}
This concludes the proof for $n=2, r=0$ and $w=0$. The case $n=2$, $r=0$ and $w=2$ is treated similarly. Extension to general momenta is again achieved via the Taylor formula (\ref{taylorfor}). Note that compared to the proof of Theorem 1, we don't need extra bounds for the cases $n=2,~r=2$ and $n=2,~r=1$ since they are not required to close the inductive scheme. 
The bounds (\ref{5})-(\ref{55}) leading to convergence are obtained 
using only the FE together with the inductive hypotheses 
(\ref{cc1}) and (\ref{ccc1}) in addition to the bound (\ref{c1}).\\
Thus, the proof of Theorem 2 is complete.
\end{proof}
Integration of the bounds (\ref{cc1}) and (\ref{ccc1}) over ${\Lambda_0}$ immediately proves the convergence of all 
$\mathcal{L}_{l,n}^{\Lambda,\Lambda_0}(z_1,\vec{p}_n;\phi_{\tau,y_{2,s}})$ for fixed $\Lambda$ to finite limits when $\Lambda_0\rightarrow \infty$. In particular, one obtains 
for all $\Lambda_0'>\Lambda_0$ and $\vec{p}_n\in \mathbb{R}^{3n}\,$
\begin{eqnarray}\label{cauchy}
\fl \left | \mathcal{L}_{l,n}^{0,\Lambda_0} (z_1;\vec{p}_n;\phi_{\tau,y_{2,s}}) - \mathcal{L}_{l,n}^{0,\Lambda_0'}(z_1;\vec{p}_n;\phi_{\tau,y_{2,s}}) \right | \nonumber\\<  \frac{m^{5-n}}{\Lambda_0}\left(\log\frac{\Lambda_0+m}{m}\right)^{\nu}\tilde{\mathcal{P}}_2\left( \frac{\left\|\vec{p}_n\right\|}{m} \right)\tilde{\mathcal{Q}}_1\left(\frac{\tau^{-\frac{1}{2}}}{m}\right)\mathcal{F}^{0}_{s,l}(\tau)\ .
\end{eqnarray}
Then the Cauchy criterion in $\mathcal{C}^{\infty}(\mathbb{R}^+)$ 
w.r.t. to $\Lambda_0$ implies the existence of finite 
limits to all loop orders $l\,$.

\newpage
\section*{Appendix A: The heat kernel $p_B$: Properties and bounds}
Here we collect inequalities verified by the one-dimensional heat kernel 
defined by  
\begin{equation}\label{PB}
p_B\left(\tau;z_1,z_2\right)=\frac{1}{\sqrt{2\pi\tau}}e^{-\frac{\left(z_1-z_2\right)^2}{2\tau}},~~\tau>0\ .
\end{equation}
Clearly, we have the following properties of (\ref{PB})
\begin{itemize}
    \item Normalization:
    \begin{equation}\label{intc}
        \int_{\mathbb{R}}du~p_B(\tau;z,u)=1,~~~\forall z\in \mathbb{R}\ .
    \end{equation}
    \item The semi-group property:
    \begin{equation}\label{sgr}
        \int_{\mathbb{R}}du~ p_B(\tau_1;z_1,u)~p_B(\tau_2;u,z_2)
=p_B(\tau_1+\tau_2;z_1,z_2),~~~~~\forall z_1,z_2\in \mathbb{R}\ .
    \end{equation}
    \item Let $p_R$ be the Robin heat kernel given by 
    $$p_R(\tau;z,z')=p_B(\tau;z,z')+\, p_B(\tau;z,-z')
-2\int_0^{\infty}\!\! dw~ e^{-w}\ 
p_B\left(\tau;z,-\frac{w}{c}-z'\right)~~~~\forall z, z'
\in \mathbb{R}^+\ .$$
   Then 
    \begin{equation}\label{prb}
        p_R(\tau;z,z')\leq 2~p_B(\tau;z,z')
~~~~\forall z, z'\in \mathbb{R}^+\ .
    \end{equation}
    \item A simple 
computation gives that for all $z_1$ and $z_2$ in $\mathbb{R}^+$, we have 
    \begin{equation}\label{rr+}
       \fl  \int_{\mathbb{R}}du~ p_B(\tau_1;z_1,u)
~p_B(\tau_2;u,z_2)\leq 2 \int_{\mathbb{R}^+}du~ p_B(\tau_1;z_1,u)
~p_B(\tau_2;u,z_2)~.
    \end{equation}
    and
    \begin{equation}\label{in1}
    \fl \indent |z_1-z_2|^r\ 
p_B\left(\tau;z_1,z_2\right)
\leq C\  {\tau^{\frac{r}{2}}}p_B\left(\tau;z_1,z_2
\right),~~C=2^{\frac{r+1}{2}} \, \| \,x^r e^{-\frac{x^2}{2}}\|_{\infty}\ ,
\end{equation}
\item For $0<\delta<1$ and $z_1,z_2\in \mathbb{R}^+$, we have 
\begin{eqnarray}\label{pbdelta}
     p_B\left(\tau;z_1,z_2\right)\leq \sqrt{1+\delta}~p_B\left(\tau_{\delta};z_1,z_2\right).
\end{eqnarray}
\end{itemize}
The following lemmata are used repeatedly in the inductive proof 
of the bounds (\ref{c1})-(\ref{c3}).
\begin{lemma}\label{lemma2}
For all $t$, $u\,$, $v\,$ and $y$ in $\mathbb{R}^+$, $\tau>0$ 
and some constant $C_{k,\delta}>0$
\begin{equation*}
\left|{\partial_t^k}p_B\left(\tau;tu+(1-t)v,y\right)\right|
\leq C_{k,\delta} \frac{|u-v|^k}{\tau^{\frac{k}{2}}} {p}_B
\left(\tau_{\delta};tu+(1-t)v,y\right)\ .
\end{equation*}
\end{lemma}
\begin{proof}
One can prove by induction that
\begin{equation*}
   \fl \indent {\partial_t^k}p_B\left(\tau;tu+(1-t)v,y\right)
=\frac{(u-v)^k}{\tau^{\frac{k}{2}}}\,
\mathcal{P}_k\left(\frac{tu+(1-t)v-y}{\sqrt{\tau}}\right)
\, p_B(\tau;tu+(1-t)v,y)~.
\end{equation*}
$\mathcal{P}_k$ is a polynomial  of degree $k$ and has at least one root if $k$ is odd. Therefore,
\begin{eqnarray*}
   \fl \indent 
\left|{\partial_t^k}p_B\left(\tau;tu+(1-t)v,y\right) \right|
\leq \left|\mathcal{P}_k\left(\frac{tu+(1-t)v-y}{\sqrt{\tau}}\right)\right|
\, e^{-\frac{\left(tu+(1-t)v-y\right)^2}{2(1+\delta)\tau}\cdot\frac{\delta}{1+2\delta}}\\
\indent \times\ 
\frac{|u-v|^k}{\tau^{\frac{k}{2}}}p_B(\tau_{\delta};tu+(1-t)v,y)\ .
\end{eqnarray*}
The lemma follows directly with 
$\,C_{k,\delta}:=\sup_{x\in \mathbb{R}}\left|\mathcal{P}_k(x)e^{-\frac{x^2}{1+\delta}\cdot\frac{\delta}{1+2\delta}}\right|$~.
\end{proof}
\begin{corollary}
For all $t$, $u\,$, $v\,$ and $y$ in $\mathbb{R}^+$, $\tau>0$ 
and some constant $C'_{k,\delta}>0$
\begin{equation*}
\left|{\partial_t^k}p_R\left(\tau;tu+(1-t)v,y\right)\right|
\leq C'_{k,\delta} \frac{|u-v|^k}{\tau^{\frac{k}{2}}} 
{p}_B\left(\tau_{\delta};tu+(1-t)v,y\right)\ .
\end{equation*}
\end{corollary}
\begin{proof}
We have 
$$p_R(\tau;z,z')=\,
p_B(\tau;z,z')+p_B(\tau;z,-z')
-2\int_0^{\infty}dw~ e^{-w}p_B\left(\tau;z,-\frac{w}{c}-z'\right)\ .$$
Using Lemma \ref{lemma2} we obtain 
\begin{eqnarray*}
\fl \left|\partial_t^k p_R\left(\tau;tu+(1-t)v,y\right)\right|
\leq C_{k,\delta}\frac{|u-v|^k}{\tau^{\frac{k}{2}}}\left({p}_B
\left(\tau_{\delta};tu+(1-t)v,y\right)\right.\\
\fl~~~\left.+~{p}_B
\left(\tau_{\delta};tu+(1-t)v,-y\right)+2\int_0^{\infty}dw~e^{-w}~
{p}_B\left(\tau_{\delta};tu+(1-t)v,-y-\frac{w}{c}\right)\right)\ .
\end{eqnarray*}
Using that for $z,~z',~w\in\mathbb{R}^+$
\begin{equation*}
    e^{-\frac{(z+z'+w/c)^2}{2}}\leq e^{-\frac{(z+z')^2}{2}}\leq e^{-\frac{(z-z')^2}{2}}
\end{equation*}
we obtain 
\begin{eqnarray*}
 {\partial_t^k}p_R\left(\tau;tu+(1-t)v,y\right)
\leq 4~C_{k,\delta}\frac{|u-v|^k}{\tau^{\frac{k}{2}}}{p}_B
\left(\tau_{\delta};tu+(1-t)v,y\right).
\end{eqnarray*}
We may set $C'_{k,\delta}:=4C_{k,\delta}$.
\end{proof}
 \begin{lemma}\label{lemme1}
Let $0<\delta'<1$, $\Lambda_I\geq \Lambda$ and $b>0$ such that $\delta'\Lambda^2 \geq b\tau^{-1}$. For $z_1,\,z_2\,$  
and $y_1$ in $\mathbb{R}^+$ and $\tau>0$ we have
 \begin{eqnarray}\label{37}
    \fl  \left|z_1-z_2\right|p_{B}\left(\frac{1+\delta}{\Lambda_{I}^2};z_1,z_2\right)\int_0^1 
dt~p_{B}\left(\tau_{\delta'};t z_2+(1-t)z_1,y_1\right)\nonumber\\ \leq C_{\delta}~  \Lambda^{-1}
p_{B}\left(\frac{2}{\Lambda_{I}^2};z_1,z_2\right)
p_{B}\left((1+\delta')^3\tau;z_1,y_1\right)\ .
 \end{eqnarray}
 \end{lemma}
 \begin{proof}
 We have 
 \begin{equation*}
     p_{B}\left(\frac{1+\delta}{\Lambda_{I}^2};z_1,z_2\right)=\sqrt{\frac{1+2\delta}{1+\delta}} ~~p_{B}\left(\frac{1+2\delta}{\Lambda_{I}^2};z_1,z_2\right)e^{-\frac{\Lambda_I^2\left(z_1-z_2\right)^2}{2(1+\delta)}\cdot\frac{\delta}{1+2\delta}}~,
 \end{equation*}
 which implies 
 \begin{eqnarray}
    \fl\indent  \left|z_1-z_2\right| p_{B}\left(\frac{1+\delta}{\Lambda_{I}^2};z_1,z_2\right)
\leq C_{\delta}~\Lambda^{-1}~p_{B}\left(\frac{1+2\delta}{\Lambda_{I}^2};z_1,z_2\right),\label{few} 
 \end{eqnarray}
where $C_{\delta}:=\sqrt{\frac{1+2\delta}{1+\delta}}\left\|x~e^{-\frac{x^2}{2(1+\delta)}\cdot\frac{\delta}{1+2\delta}}\right\|_{\infty}$, and we used that $\Lambda\leq \Lambda_I$.\\
Now, we bound 
\begin{equation}\label{feeww}
   p_{B}\left(\frac{1+2\delta}{\Lambda_{I}^2};z_1,z_2\right)\int_0^1 
dt~p_{B}\left(\tau_{\delta'};t z_2+(1-t)z_1,y_1\right)~.
\end{equation}
 For  $0<\delta'< 1$ and $\delta'\Lambda_I^2\geq b\tau^{-1}$, we have
\begin{eqnarray}
     \fl \frac{\Lambda_{I}^2}{1+2\delta}(z_1-z_2)^2&+\frac{1}{\tau(1+\delta')}(t z_2+(1-t)z_1-y_1)^2\\&=\frac{\Lambda_I^2|z_1-z_2|^2}{b}+\frac{\Lambda_I^2|z_1-z_2|^2}{2}+\frac{1}{\tau(1+\delta')}(t z_2+(1-t)z_1-y_1)^2\nonumber\\
&\geq\frac{|z_1-z_2|^2}{\tau\delta'}+\frac{\Lambda_I^2|z_1-z_2|^2}{4}+\frac{1}{\tau(1+\delta')^2}(t z_2+(1-t)z_1-y_1)^2~.\nonumber
\end{eqnarray}
Let $0\leq t\leq 1$, we have
\begin{eqnarray}
\fl \frac{|z_1-z_2|^2}{\delta'}+\frac{|t z_2+(1-t)z_1-y_1|^2}{(1+\delta')^2}&\geq
\frac{1}{(1+\delta')^2}\left[\frac{1}{\delta'}|z_1-z_2|^2+|tz_2+(1-t)z_1-y_1|^2\right]~.\nonumber
 \end{eqnarray}
 Using that 
 \begin{equation*}
     |z_1-z_2|=|z_1-tz_2-(1-t)z_1|+|tz_2+(1-t)z_1-z_2|~,
 \end{equation*}
 we obtain 
 \begin{eqnarray*}
      \fl~~\frac{1}{\delta'}|z_1-z_2|^2&+|tz_2+(1-t)z_1-y_1|^2\nonumber\\
      &\geq \frac{1}{\delta'}|z_1-tz_2-(1-t)z_1|^2+|tz_2+(1-t)z_1-y_1|^2\\
      &\geq \frac{1}{1+\delta'}\left(|z_1-tz_2-(1-t)z_1|+|tz_2+(1-t)z_1-y_1|\right)^2\\
      &\geq \frac{|z_1-y_1|^2}{1+\delta'}~.
 \end{eqnarray*}
Therefore we obtain that 
\begin{eqnarray}
     \fl~~~ \frac{\Lambda_{I}^2}{1+2\delta}(z_1-z_2)^2&+\frac{1}{\tau(1+\delta')}(t z_2+(1-t)z_1-y_1)^2\geq\frac{|z_1-y_1|^2}{\tau(1+\delta')^3}+\frac{\Lambda_I^2|z_1-z_2|^2}{4}~,\nonumber
\end{eqnarray}
which implies that (\ref{feeww}) can be  bounded by 
 \begin{equation*}
     C_{\delta}~\Lambda^{-1}~ p_{B}\left(\frac{2}{\Lambda_{I}^2};z_1,z_2\right)p_{B}\left((1+\delta')^3\tau;z_1,y_1\right),
 \end{equation*}
 which ends the proof.
 \end{proof}


\newpage
\section*{References}
\begin{thebibliography}{}
\bibitem{1} 
K. Symanzik, "Schrödinger representation in renormalizable quantum field theory", Nucl. Phys. B190 (1981) 
\bibitem{3} I. Gelfand and N. Vilenkin, "Generalized functions", tome IV, p. 329, New York: Academic Press (1964)
\bibitem{4} T. Hida, "Stationary stochastic processes", Princeton University Press (1970)
\bibitem{5} J. Potthoff, "On Differential Operators in White Noise Analysis", Acta Applicandae Mathematicae 63, 333–347 (2000)  
\bibitem{26} J. Glimm and A. Jaffe,"Quantum Physics: A Functional Integral Point of View", Springer Verlag, p. 137 (1987)
\bibitem{8}
G. Barton, "Elements of Green’s Functions and Propagation", Oxford Science publications, p. 33 (1989).
\bibitem{2} D. L. Mills, "Surface Effects in Magnetic Crystals near the Ordering Temperature", Phys. Rev. B3, 3887-3895 (1971) 
\bibitem{17} T. C. Lubensky and M. H. Rubin, "Critical phenomena in semi-infinite systems. I. $\epsilon$ expansion for positive extrapolation length", Phys. Rev. B11, 4533-4546, (1975) 
\bibitem{18} T. C. Lubensky and M. H. Rubin, "Critical phenomena in semi-infinite systems. II. Mean-field theory", Phys. Rev. B12, 3885-3901 (1975) 
\bibitem{22} K. Binder and P. C. Hohenberg, "Surface effects on magnetic phase transitions", Phys. Rev. B9, 2194-2214 (1974) 
\bibitem{23} A. J. Bray and M. A. Moore, "Surface Critical Exponents in Terms of Bulk Exponents", Phys. Rev. Lett. 38, 1046-1048 (1977) 
\bibitem{24} A. J. Bray and M. A. Moore, "Critical Behavior of a Semi-infinite System: n-Vector Model in the Large-n Limit", Phys. Rev. Lett. 38, 735-738 (1977) 
\bibitem{7} H. W. Diehl and S. Dietrich, "Field-theoretical approach to multicritical behavior near free surfaces", Phys. Rev. B24, 2878-2880 (1981) 
\bibitem{8} M. Reed and L. Rosen, "Support properties of the free measure for Boson fields", Commun. Math. Phys. 36, 123–132 (1974) 
\bibitem{9} Ch. Kopper, V. F. Müller, "Renormalization Proof for Massive $\phi_4^4$ Theory on Riemannian Manifolds", Commun. Math. Phys. 275, 331–372 (2007) 
\bibitem{25} M. Borji and Ch. Kopper, "Perturbative renormalization of the lattice regularized $\phi_4^4$ with flow equations", J. Math. Phys. 61, 112304 (2020)
\bibitem{16} G. Keller, Ch. Kopper and M. Salmhofer, "Perturbative renormalization and effective
Lagrangians in $\phi^4$ in four-dimensions", Helv. Phys. Acta 65, 32-52 (1992) 
\bibitem{17}
G. Keller, Ch. Kopper and C. Schophaus, "Perturbative renormalization with flow
equations in Minkowski space",
Helv. Phys. Acta 70, 247-274 (1997) 
\bibitem{18}
V. F. Müller, "Perturbative renormalization by flow equations", Rev. Math. Phys. 15, 491-558 (2003)
\bibitem{19}
J. Polchinski, "Renormalization and Effective Lagrangians", Nucl. Phys. B231, 269-295 (1984) 
\bibitem{10} H. W. Diehl and S. Dietrich, "Field-theoretical approach to static critical phenomena in semi-infinite systems", Z. Phys. B Condensed Matter 42, 65–86 (1981)
\bibitem{11} H. W. Diehl, "The Theory of Boundary Critical Phenomena", International Journal of Modern Physics B, Volume 11, Issue 30,  3503-3523 (1997)
\bibitem{12} H. W. Diehl, S. Dietrich, "Multicritical behaviour at surfaces", Z. Phys. B, Condensed Matter 50, 117–129 (1983) 
\bibitem{13} H. W. Diehl, "Why boundary conditions do not generally determine the universality class for boundary critical behavior", Eur. Phys. J. B 93  (10):195 (2020)
\bibitem{27} H. W. Diehl, in "Phase Transitions and Critical Phenomena",
edited by C. Domb and J. L. Lebowitz,  Academic Press London,
Vol. 10, 75–267 (1986).
\bibitem{14} L. C. de Albuquerque, "Renormalization of the $\phi_4^4$ scalar theory under Robin boundary conditions and a possible new renormalization ambiguity",  	arXiv:hep-th/0507019
\bibitem{15} L. C. de Albuquerque and R. M. Cavalcanti, "Casimir effect for the scalar field under Robin boundary conditions: A functional integral approach", J. Phys. A: Math. Gen. 37  7039-7050 (2004)
\end{thebibliography}
\end{document}